\documentclass[11pt]{article}
\usepackage{longtable}
\usepackage{ae,aecompl,pslatex}
\usepackage[dvips]{graphicx}
\usepackage{enumerate,color}
\usepackage{fullpage}
\usepackage{lineno}

\usepackage{placeins}
\usepackage{psfrag,epsfig,graphicx,color,array,dcolumn,tabularx,delarray,longtable,indentfirst,amsmath,amsfonts,amssymb,amsthm,eucal,bbm}
\theoremstyle{bkaexa} 
\usepackage{caption}
\usepackage{setspace}
\theoremstyle{bkaexa} 

\theoremstyle{bkathm} 

\theoremstyle{bkathm} 
\newtheorem{Thm}{Theorem}
\theoremstyle{bkathm} 

\theoremstyle{bkathm} 
\newtheorem{Lem}{Lemma}
\theoremstyle{definition}

\begin{document}
\setstretch{1.5}
\title{On the properties of distance covariance for categorical data: Robustness, sure screening, and approximate null distributions}
\author{\normalsize Qingyang Zhang\\
\normalsize Department of Mathematical Sciences, University of Arkansas, AR 72701\\
\normalsize Email: qz008@uark.edu
}
\date{}
\maketitle

\begin{abstract}
Pearson's Chi-squared test, though widely used for detecting association between categorical variables, exhibits low statistical power in large sparse contingency tables. To address this limitation, two novel permutation tests have been recently developed: the distance covariance permutation test and the U-statistic permutation test. Both leverage the distance covariance functional but employ different estimators. In this work, we explore key statistical properties of the distance covariance for categorical variables. Firstly, we show that unlike Chi-squared, the distance covariance functional is B-robust for any number of categories (fixed or diverging). Second, we establish the strong consistency of distance covariance screening under mild conditions, and simulations confirm its advantage over Chi-squared screening, especially for large sparse tables. Finally, we derive an approximate null distribution for a bias-corrected distance correlation estimate, demonstrating its effectiveness through simulations.
\end{abstract}


\section{Introduction}
Detecting associations between categorical variables is essential in various scientific fields. In genetics, for example, researchers often investigate the linkage between genotype or haplotype and phenotypic traits. Similarly, social scientists frequently examine the relationships between various sociodemographic factors, such as socioeconomic status, educational attainment, marital status, and political affiliation. Pearson's Chi-squared test has been a widely used tool for this purpose. However, its statistical power may deteriorate when applied to relatively large and sparse contingency tables, as the underlying assumptions of the Chi-squared test can be severely violated. To address this limitation, two novel permutation tests, including the distance covariance permutation test \cite{Zhang2019} and U-statistic permutation test \cite{USP, USPAOS}, have been recently developed. These two tests demonstrate greater power than Pearson's Chi-squared test, particularly when applied to  large contingency tables with limited sample sizes (e.g., see Figure 2 in \cite{Zhang2019} and Figure 3 in \cite{USP}). Notably, both tests leverage an $L_{2}$-type functional (equivalent to the distance covariance functional) but employ different estimators. We begin with a brief review of Pearson's Chi-squared test and the two permutation tests.

Let $X$ and $Y$ be two categorical variables, with $X\in\{1,~ ...,~ I\}$ and $Y\in\{1,~ ...,~ J\}$. Let $\pi = \{\pi_{ij}\}_{1\leq i\leq I, 1\leq j\leq J}$ be the joint distribution of $(X, Y)$, where $\pi_{ij} = P(X = i, Y = j)$, and $\{\pi_{i+}\}_{1\leq i\leq I}$ and $\{\pi_{+j}\}_{1\leq j\leq J}$ be the marginal probabilities. Let $n_{ij}$ be the observed count in cell $(i, j)$, $n = \sum_{i=1}^{I}\sum_{j=1}^{J}n_{ij}$ be the total sample size, $n_{i+} = \sum_{j=1}^{J}n_{ij}$ and $n_{+j} = \sum_{i=1}^{I}n_{ij}$ be the marginal counts. The maximum likelihood estimates are $\widehat{\pi}_{ij} = n_{ij}/n$, $\widehat{\pi}_{i+} = n_{i+}/n$, and $\widehat{\pi}_{+j} = n_{+j}/n$. The hypothesis test for independence between $X$ and $Y$ can be formulated as 
\begin{align*}
H_{0}: \pi_{ij} &= \pi_{i+}\pi_{+j}\text{~for~}1\leq i\leq I,~ 1\leq j\leq J, \\
H_{a}: \pi_{ij} &\neq \pi_{i+}\pi_{+j}\text{~for~some~}(i, j),
\end{align*}
and Pearson's Chi-squared statistic can be expressed as
\begin{equation*}
\widehat{\eta} = \sum_{i=1}^{I}\sum_{j=1}^{J}\frac{(\widehat{\pi}_{ij} - \widehat{\pi}_{i+}\widehat{\pi}_{+j})^2}{\widehat{\pi}_{i+}\widehat{\pi}_{+j}},
\end{equation*} 
which is based on the following functional 
\begin{equation*}
\eta = \sum_{i=1}^{I}\sum_{j=1}^{J}\frac{(\pi_{ij} - \pi_{i+}\pi_{+j})^2}{\pi_{i+}\pi_{+j}}.
\end{equation*} 
A closely related quantity is the likelihood ratio Chi-squared statistics, defined as
\begin{equation*}
\widehat{g} = 2\sum_{i=1}^{I}\sum_{j=1}^{J} \widehat{\pi}_{ij} \log\frac{\widehat{\pi}_{ij}}{\widehat{\pi}_{i+}\widehat{\pi}_{+j}}.
\end{equation*} 
It is well known that $\widehat{\eta}$ is the second-order Taylor expansion of $\widehat{g}$, thus in general $\widehat{\eta}\approx\widehat{g}$. Under independence and the normal condition $\min_{i, j}n\widehat{\pi}_{i+}\widehat{\pi}_{+j} \geq 5$, both $n\widehat{\eta}$ and $n\widehat{g}$ are close to a Chi-squared distribution with $df = (I-1)(J-1)$. Empirically, the approximation is still satisfactory even when up to 20\% of cells violate the normal condition. However, in large sparse tables, both tests might miss the true associations due to limited sample sizes. 

The two aforementioned permutation tests are both based on the following $L_2$-type functional 
\begin{equation*}
\Delta = \sum_{i=1}^{I}\sum_{j=1}^{J}(\pi_{ij} - \pi_{i+}\pi_{+j})^2,
\end{equation*} 
which is equivalent to the squared distance covariance between $X$ and $Y$ \cite{Zhang2019, DC}. Zhang (2019) considered the following maximum likelihood estimate for permutation test
\begin{equation*}
\widehat{\Delta} = \sum_{i=1}^{I}\sum_{j=1}^{J}\left(\widehat{\pi}_{ij} - \widehat{\pi}_{i+}\widehat{\pi}_{+j} \right)^2,
\end{equation*} 
while Berrett \& Samworth (2021) used an unbiased estimate derived from the fourth-order U-statistic. To illustrate, let $\{(X_{1},~ Y_{1}),~ ...,~ (X_{n},~ Y_{n})\}$ be i.i.d. samples. The unbiased estimate of squared distance covariance can be written as
\begin{equation*}
\widetilde{\Delta} = \frac{T_{1}}{n(n-3)} - \frac{2T_{2}}{n(n-2)(n-3)} + \frac{T_{3}}{n(n-1)(n-2)(n-3)},
\end{equation*}
where
\begin{align*}
T_{1} & = \sum_{l=1}^{n}\sum_{m=1}^{n} \| X_{l} - X_{m} \|\cdot \| Y_{l} - Y_{m} \|, \\
T_{2} & = \sum_{l=1}^{n}\left( \sum_{m=1}^{n}  \| X_{l} - X_{m} \| \sum_{m=1}^{n}  \| Y_{l} - Y_{m} \| \right), \\
T_{3} & = \left(\sum_{l=1}^{n}\sum_{m=1}^{n} \| X_{l} - X_{m} \|\right) \left(\sum_{l=1}^{n}\sum_{m=1}^{n} \| Y_{l} - Y_{m} \|\right). 
\end{align*}
 
For categorical variables $X$ and $Y$, let $\| X_{l} - X_{m} \| = 0$ if $X_{l} = X_{m}$ and 1 otherwise, $\widetilde{\Delta}$ can be expressed as
\begin{align*}
\widetilde{\Delta} & = \frac{n}{n-3} \sum_{i = 1}^{I} \sum_{j=1}^{J} (\widehat{\pi}_{ij} - \widehat{\pi}_{i+}\widehat{\pi}_{+j})^{2} - \frac{4n}{(n-2)(n-3)}\sum_{i = 1}^{I} \sum_{j=1}^{J}\widehat{\pi}_{ij}\widehat{\pi}_{i+}\widehat{\pi}_{+j}  \\
& + \frac{n}{(n-1)(n-3)}\left( \sum_{i = 1}^{I}\widehat{\pi}_{i+}^{2} + \sum_{j=1}^{J}\widehat{\pi}_{+j}^{2} \right) + \frac{n(3n-2)}{(n-1)(n-2)(n-3)}\left(\sum_{i = 1}^{I}\widehat{\pi}_{i+}^{2}\right)\left(\sum_{j=1}^{J}\widehat{\pi}_{+j}^{2}\right)  \\
& - \frac{n}{(n-1)(n-3)},
\end{align*}
and it can be shown that the maximum likelihood estimate ($\widehat{\Delta}$) and the unbiased estimate ($\widetilde{\Delta}$) are asymptotically equivalent, as $\widetilde{\Delta} - \widehat{\Delta}  = O(1/n)$ (details are given in the proof of Theorem 3). 

While simulations confirm the sensitivity of both tests to dependence departures, the properties of the distance covariance functional itself remains unexplored. This work delves into the statistical properties of $\Delta$ and its estimates. We establish the B-robustness of $\Delta$, the strong screening consistency of $\widetilde{\Delta}$ and $\widehat{\Delta}$, and derive an approximate null distribution for the normalized $\widetilde{\Delta}$. These nice properties suggest that this novel functional holds great potential for wider application, particularly in high-dimension data analysis.

The remainder of this paper is structured as follows: In Section 2, we derive the influence functions of $\eta$ and $\Delta$, and analyze the B-robustness of the two functionals. In Section 3, we focus on the problem of feature screening, and establish the strong screening consistency of both $\widehat{\Delta}$ and $\widetilde{\Delta}$. Section 4 presents an approximate null distribution for the normalized $\widetilde{\Delta}$ using a weighted sum of Chi-squared distributions, enabling analytical p-value calculation. Section 5 discusses the paper with some future perspectives.

\section{Robustness under fixed and diverging dimensions} 
In this section, we study the robustness of Pearson's Chi-squared and distance covariance functionals. With the aid of influence function \cite{Hampel74, Hampel86}, we show that the distance covariance is robust for any two categorical variables, while the Chi-squared is generally not. For two random variables with joint distribution $\pi$, the influence function ($IF$) of a statistical functional $R(\pi)$ at $(X = x, Y = y)$ is defined as 
\begin{equation*}
\mbox{IF}[(x,~ y),~ R,~ \pi)] = \lim_{\epsilon\downarrow 0}\frac{R[(1-\epsilon)\pi + \epsilon\delta_{(x,~ y)}] - R(\pi)}{\epsilon},
\end{equation*}
where $\delta_{(x,~ y)}$ represents a Dirac distribution which puts all its mass on $(x,~ y)$. The influence function above is essentially a G\^{a}teaux derivative at $(x,~ y)$, which quantifies the sensitivity of $R$ to infinitesimal modifications in a single data point $(x,~ y)$. A substantial influence function indicates that the functional is highly susceptible, whereas a modest influence function suggests that the functional is more robust to $(x,~ y)$. The gross error sensitivity can be used to measure the maximum change to $R$ that a small perturbation to $\pi$ at a point can induce, expressed as 
\begin{equation*}
\gamma(R,~ \pi) = \sup_{x,~ y}\left|\mbox{IF}[(x,~ y),~ R,~ \pi)]\right|,
\end{equation*}
and $R$ is said to be B-robust at $\pi$ if $\gamma(R,~ \pi)$ is bounded. 

Lemma 1 below gives the influence function of $\Delta$ (detailed proof is provided in A.1).
\begin{Lem}
Suppose $X$ and $Y$ are categorical variables with $I$ and $J$ categories and joint distribution $\pi$, the squared distance covariance $\Delta(\pi)$ has the following influence function
\begin{align*}
\mbox{IF}[(x,~ y),~ \Delta,~ \pi)] = & 2\sum_{i=1}^{I}\sum_{j=1}^{J} (\pi_{ij} - \pi_{i+}\pi_{+j})(2\pi_{i+}\pi_{+j} - \pi_{ij}) \\
& - 2\sum_{j=1}^{J} \pi_{+j}(\pi_{xj} - \pi_{x+}\pi_{+j}) - 2\sum_{i=1}^{I} \pi_{i+}(\pi_{iy} - \pi_{i+}\pi_{+y}) \\
& + 2(\pi_{xy} - \pi_{x+}\pi_{+y}).
\end{align*}
\end{Lem}

Theorem 1 below suggests that $\Delta$ is B-robust for any two categorical variables, with fixed or diverging number of categories.
\begin{Thm}
For any fixed or diverging $(I,~ J)$ and any $\pi$, the squared distance covariance $\Delta(\pi)$ is B-robust with gross error sensitivity less than 11.
\end{Thm}
\begin{proof}
By Lemma 1, we decompose $\mbox{IF}[(x,~ y),~ \Delta,~ \pi)]$ into four parts, and then bound each part separately. Let $\mbox{IF}[(x,~ y),~ \Delta,~ \pi)] = T_{1} + T_{2} + T_{3} + T_{4}$, where 

\begin{align*}
T_{1} & = 2\sum_{i=1}^{I}\sum_{j=1}^{J} (\pi_{ij} - \pi_{i+}\pi_{+j})(2\pi_{i+}\pi_{+j} - \pi_{ij}),\\
T_{2} & = - 2\sum_{j=1}^{J} \pi_{+j}(\pi_{xj} - \pi_{x+}\pi_{+j}), \\
T_{3} & = - 2\sum_{i=1}^{I} \pi_{i+}(\pi_{iy} - \pi_{i+}\pi_{+y}), \\
T_{4} & = 2(\pi_{xy} - \pi_{x+}\pi_{+y}).
\end{align*}
First, it is easy to see that $|T_{4}|\leq 1$. For $T_2$, we have 
\begin{align*}
T_{2} & =  -2\sum_{j=1}^{J} \pi_{+j}\pi_{xj} + 2\sum_{j=1}^{J}\pi_{x+}\pi^{2}_{+j}, \\
-2 & < -2\sum_{j=1}^{J} \pi_{+j}\pi_{xj} < 0, \\
0 & < 2\sum_{j=1}^{J}\pi_{x+}\pi^{2}_{+j} < 2.
\end{align*}
Therefore $|T_{2}|<2$. Similarly, $|T_{3}| < 2$. 

Lastly, for $T_{1}$, we have
\begin{align*}
T_{1} & = \sum_{i=1}^{I}\sum_{j=1}^{J} (6\pi_{ij}\pi_{i+}\pi_{+j} - 2\pi_{i+}^{2}\pi_{+j}^{2} - 2\pi_{ij}^{2}),\\
0 & < \sum_{i=1}^{I}\sum_{j=1}^{J} 6\pi_{ij}\pi_{i+}\pi_{+j} < 6, \\
-2 & < - 2\sum_{i=1}^{I}\sum_{j=1}^{J}\pi_{i+}^{2}\pi_{+j}^{2} < 0, \\
-2 & < -2\sum_{i=1}^{I}\sum_{j=1}^{J} \pi_{ij}^{2} < 0.
\end{align*}
Therefore $|T_{1}|<6$, and $\left|\mbox{IF}[(x, y),~ \Delta,~ \pi)]\right|\leq |T_{1}| + |T_{2}| + |T_{3}| + |T_{4}| <11$ for any $I$, $J$ and $\pi$.
\end{proof}

To evaluate the robustness of the Chi-squared functional ($\eta$), Lemma 2 below gives its influence function (proof is given in A.2).
\begin{Lem}
Suppose $X$ and $Y$ are categorical variables with $I$ and $J$ categories and joint distribution $\pi$, Pearson's Chi-squared functional $\eta(\pi)$ has the following influence function
\begin{equation*}
\mbox{IF}[(x,~ y),~ \eta,~ \pi)] = \sum_{i\neq x}\sum_{j\neq y}A_{ij} + \sum_{j\neq y}B_{xj} + \sum_{i\neq x} C_{iy} + D_{xy} ,
\end{equation*}
where 
\begin{align*}
A_{ij} & = 2(\pi_{ij} - 2\pi_{i+}\pi_{+j}), \\
B_{xj} & = - \frac{\pi_{xj}^{2}}{\pi_{x+}\pi_{+j}} - \frac{\pi_{xj}^{2}}{\pi_{x+}^{2}\pi_{+j}} - 3\pi_{x+} + 3\pi_{+j} + 4\pi_{xj}, \\
C_{iy} & = - \frac{\pi_{iy}^{2}}{\pi_{i+}\pi_{+y}} - \frac{\pi_{iy}^{2}}{\pi_{i+}\pi_{+y}^{2}} + 3\pi_{i+} - 3\pi_{+y} + 4\pi_{iy}, \\
D_{xy} & = \frac{\pi_{xy}^{2}}{\pi_{x+}\pi_{+y}^2} + \frac{\pi_{xy}^{2}}{\pi_{x+}^{2}\pi_{+y}} - \frac{2\pi_{xy}}{\pi_{x+}\pi_{+y}} - 2\pi_{xy} + 2\pi_{x+}\pi_{+y} - \pi_{x+} - \pi_{+y} -2.
\end{align*}
\end{Lem}

Theorem 2 below implies that $\eta(\pi)$ is generally not B-robust, unless $\pi$ satisfies certain assumptions.
\begin{Thm}
For fixed $I$ and $J$, unless there exists $\pi_{min}>0$ such that $(\min_{i}\pi_{i+})\wedge(\min_{j}\pi_{+j}) > \pi_{min}$, the Chi-squared functional $\eta(\pi)$ is generally not B-robust. For diverging $I$ and $J$, $\eta(\pi)$ is generally not B-robust.
\end{Thm}
\begin{proof}
Using notations from Lemma 2, for any $(I,~ J,~\pi)$, we have
\begin{align*}
\left|\sum_{i\neq x}\sum_{j\neq y}A_{ij}\right| & =  \left|\sum_{i=1}^{I}\sum_{j=1}^{J} (2\pi_{ij} - 4\pi_{i+}\pi_{+j}) - \sum_{j=1}^{J}(2\pi_{xj}-4\pi_{x+}\pi_{+j}) - \sum_{i=1}^{I}(2\pi_{iy}-4\pi_{i+}\pi_{+y}) + (2\pi_{xy} - 4\pi_{x+}\pi_{+y})\right|\\
& = \left|-2 + 2\pi_{x+} + 2\pi_{+y} + 2\pi_{xy} - 4\pi_{x+}\pi_{+y}\right| \\
& < 6.
\end{align*}
We first show that for fixed $I$ and $J$, if there exists $\pi_{min}>0$ such that $(\min_{i}\pi_{i+})\wedge(\min_{j}\pi_{+j}) > \pi_{min}$, then $\eta(\pi)$ is B-robust. For $B_{xj}$, we have 
\begin{align*}
\left|\sum_{j\neq y}B_{xj}\right| & <  \left|\sum_{j\neq y} \frac{\pi_{xj}^{2}}{\pi_{x+}\pi_{+j}}\right| + \left|\sum_{j\neq y}\frac{\pi_{xj}^{2}}{\pi_{x+}^{2}\pi_{+j}}\right| +\left| \sum_{j\neq y} (3\pi_{+j} - 3\pi_{x+} + 4\pi_{xj}) \right| \\
& < J + \frac{J}{\pi_{min}} + 3 + 7J \\
& < \frac{J}{\pi_{min}} + 8J + 3.
\end{align*}
Similarly, for $C_{iy}$ and $D_{xy}$, we have 
\begin{equation*}
\left|\sum_{i\neq x} C_{iy}\right| < \frac{I}{\pi_{min}} + 8I + 3,
\end{equation*}
and
\begin{equation*}
\left| D_{xy}\right| < \frac{4}{\pi_{min}} + 8.
\end{equation*}
Therefore $\eta(\pi)$ is B-robust. 

Next we show that for fixed $I$ and $J$, $\eta(\pi)$ is not always B-robust without the condition on $(\min_{i}\pi_{i+})\wedge(\min_{j}\pi_{+j})$. We give a counterexample here. Let $c$ be a positive constant such that 
$$(\min_{i\neq x}\pi_{i+})\wedge(\min_{j\neq y}\pi_{+j})>c.$$ 
Here it is noteworthy that $c$ is different from $\pi_{min}$, as it is for $i\neq x$ and $j\neq y$. In addition, let $\pi_{xy} = \beta$, $\pi_{x+} = \pi_{+y} = 2\beta$, $\pi_{xj} = \beta/(J-1)$ for $j\neq y$, $\pi_{iy} = \beta/(I-1)$ for $i\neq x$, and $\beta\rightarrow 0$, then we have
\begin{align*}
\left|\sum_{j\neq y}B_{xj}\right| & < \frac{J}{c} + 8J + 3, \\
\left|\sum_{i\neq x}C_{iy}\right| & < \frac{I}{c} + 8I + 3, \\
D_{xy} = & -\frac{1}{4\beta} + 8\beta^2 -6\beta -2\rightarrow -\infty.
\end{align*}
Therefore $\eta(\pi)$ is not B-robust. 

Finally, we show that for diverging $I$ and $J$, $\eta(\pi)$ is generally not B-robust. Let $\pi_{xy} = \pi_{x+} = \pi_{+y} = \alpha>0$, $\pi_{xj} = 0$ for $j\neq y$, $\pi_{iy} = 0$ for $i\neq x$, then we have
\begin{align*}
D_{xy} & =  2\alpha^2 -4\alpha - 2, \\ 
\sum_{j\neq y}B_{xj} & = 3 - 3\alpha - 3(J-1)\alpha \rightarrow -\infty, ~\text{as}~ J\rightarrow \infty\\
\sum_{i\neq x}C_{iy} & = 3 - 3\alpha - 3(I-1)\alpha \rightarrow -\infty, ~\text{as}~ I\rightarrow \infty. 
\end{align*}
Therefore $\eta(\pi)$ is generally not B-robust.
\end{proof}
As a comparison, Theorems 1 and 2 imply that although both $\eta$ and $\Delta$ are valid dependence metrics (0 if and only if $X$ and $Y$ are independent), $\Delta$ is generally less sensitive to extreme values in $\pi$ than $\eta$. Intuitively, this is because $\eta$ is a sum of ratio statistics with $\pi_{i+}\pi_{+j}$ in the denominators, which can become unstable when $\pi_{i+}$ or $\pi_{+j}$ is very small, a common occurrence in large sparse tables. Therefore $\Delta$ is preferable to $\eta$ in terms of functional robustness.

\section{Strong screening consistency}
In this section, we consider feature screening for high dimensional categorical data, where both the response and features are categorical \cite{FanLv, FanSong, LiZhongZhu, Huangetal}. Huang et al. (2014) proposed a screening procedure based on Pearson's Chi-squared functional ($\eta$) and demonstrated its strong screening consistency (terminology of Fan \& Lv, 2008). Here we explore using distance covariance functional for feature screening. We prove that the distance covariance screening also achieves screening consistency, under conditions similar to, but slightly weaker than, those in \cite{Huangetal}. Inspired by change-point detection, we suggest a change-point based method for tuning parameter selection. Our simulation studies confirm that these methods perform well in high-dimensional settings with relatively low average cell counts. 

\subsection{Problem formulation}
We begin by introducing the notations and formulating the problem. Let $Y\in\{1,~ ...,~ J\}$ be the the response variable, and $S = \{X_{1},~ ...,~ X_{K}\}$ be the vector of $K$ features, where $X_{k} \in \{1,~ ...,~ I_{k}\}$. Let $\pi_{i_{k}j} = P(X_{k} = i_{k}, Y = j)$ and $\pi_{i_{k}+} = P(X_{k} = i_{k})$ denote the joint and marginal probabilities. Let $X_{m} = (X_{m1},~ ...,~ X_{mK})$, $m = 1,~ ...,~ n$, be $n$ independent samples, then $\pi_{i_{k}j}$ and $\pi_{i_{k}+}$ can be estimated by corresponding sample proportions, denoted by $\widehat{\pi}_{i_{k}j}$ and $\widehat{\pi}_{i_{k}+}$. Finally, we define the true model, denoted by $S_{T}$, as a non-empty subset of features in $S$, such that $X_{k}\in S_{T}$ if and only if $X_{k}$ and $Y$ are dependent.

To measure the dependence between $X_{k}$ and $Y$, define $\Delta_{k} = \sum_{i_{k} = 1}^{I_{k}} \sum_{j=1}^{J} (\pi_{i_{k}j} - \pi_{i_{k}+}\pi_{+j})^{2}$. Higher values of $\Delta_{k}$ indicate stronger dependence between $X_{k}$ and $Y$. For feature screening, we can consider the following maximum likelihood estimator ($\widehat{\Delta}_{k}$) and unbiased estimator ($\widetilde{\Delta}_{k}$) 
\begin{equation*}
\widehat{\Delta}_{k} = \sum_{i_{k} = 1}^{I_{k}} \sum_{j=1}^{J} (\widehat{\pi}_{i_{k}j} - \widehat{\pi}_{i_{k}+}\widehat{\pi}_{+j})^{2},
\end{equation*}
\begin{align*}
\widetilde{\Delta}_{k} & = \frac{n}{n-1} \sum_{i_{k} = 1}^{I_{k}} \sum_{j=1}^{J} (\widehat{\pi}_{i_{k}j} - \widehat{\pi}_{i_{k}+}\widehat{\pi}_{+j})^{2} - \frac{4n}{(n-2)(n-3)}\sum_{i_{k} = 1}^{I_{k}} \sum_{j=1}^{J}\widehat{\pi}_{i_{k}j}\widehat{\pi}_{i_{k}+}\widehat{\pi}_{+j}  \\
& + \frac{n}{(n-1)(n-3)}\left( \sum_{i_{k} = 1}^{I_{k}}\widehat{\pi}_{i_{k}+}^{2} + \sum_{j=1}^{J}\widehat{\pi}_{+j}^{2} \right) + \frac{n(3n-2)}{(n-1)(n-2)(n-3)}\left(\sum_{i_{k} = 1}^{I_{k}}\widehat{\pi}_{i_{k}+}^{2}\right)\left(\sum_{j=1}^{J}\widehat{\pi}_{+j}^{2}\right)  \\
& - \frac{n}{(n-1)(n-3)}.
\end{align*}
Obviously, features with larger values of $\widehat{\Delta}_{k}$ or $\widetilde{\Delta}_{k}$ are more likely to be relevant, therefore we can estimate the true model $S_{T}$ by $\widehat{S}(C_{1}) = \{k:~ \widehat{\Delta}_{k}>C_{1} \}$ or $\widetilde{S}(C_{2}) =  \{k:~ \widetilde{\Delta}_{k}>C_{2} \}$, where $C_{1}>0$ and $C_{2}>0$ are some predefined constants.

\subsection{Theoretical results}
We will now demonstrate that under certain conditions, the proposed methods based on $\widehat{S}(C_{1})$ and $\widetilde{S}(C_{2})$ can consistently identify the true model $S_{T}$. Let $\omega_{i_{k}j} = \mbox{cov}\{I(X_{k} = i_{k}), I(Y = j)\}$, we proceed by assuming the following conditions
\begin{itemize}
\item[(1)] There exists a positive constant $I_{max}$, such that $\max_{k}I_{k} < I_{max}$ and $J < I_{max}$.
\item[(2)] There exists a positive constant $\omega^{2}_{min}$, such that $\min_{k}\max_{i_{k}, j} \omega^{2}_{i_{k}j} > \omega^{2}_{min}$.
\item[(3)] $\log K = o(n)$.
\end{itemize}
Condition (1) assumes that the number of categories for both the response and all features are finite, as dealing with a diverging number of categories (either $J$ or $I_{k}$) in feature screening is theoretically challenging. For example, Huang et al. (2014) also assumed a finite $J$ and $I_{k}$. Conditions (2) and (3) are same as those in \cite{Huangetal}. Condition (2) requires that, for any $X_{k}\in S_{T}$, there exists at least one response category $j$ and one feature category $i_{k}$ that are correlated with marginal covariance at least $\omega^{2}_{min}$. Condition (3) requires that the feature dimension $K$ diverges no faster than the exponential of sample size. Notably, our conditions are slightly weaker than those in \cite{Huangetal}, because we relax the assumptions on response probabilities (see Condition C1 in their work). This advantage is due to the distance covariance functional not having $\pi_{+j}$ in the denominator. Consequently, it is robust to extremely small or large response probabilities.

We have the following theorem (detailed proof is given in A.3).
\begin{Thm}
Under Conditions (1)-(3), there exist positive constants $C_{1}$ and $C_{2}$, such that 
\begin{align*}
P[\widehat{S}(C_{1}) = S_{T}] \rightarrow 1, \\
P[\widetilde{S}(C_{2}) = S_{T}] \rightarrow 1,
\end{align*}
as $n\rightarrow \infty$.
\end{Thm}

\subsection{Tuning parameter selection}
Tuning parameter selection (e.g., $C_{1}$ and $C_{2}$) is crucial for many feature screening procedures. Various data-driven methods have been proposed, including the maximum ratio criterion introduced by Huang et al. (2014). Huang et al.'s method uses a sorted sequence of $\{\widehat{\eta}_{k}\}_{1\leq k\leq K}$, denoted by $\widehat{\eta}_{(K)}\geq \cdots \geq \widehat{\eta}_{(1)}$, to identify a threshold ($C$) based on the largest ratio between consecutive elements
\begin{equation*}
C = \left\{\widehat{\eta}_{(k+1)},~ k = \underset{K-1\geq r\geq 1}{\arg\max} \frac{\widehat{\eta}_{(r+1)}}{\widehat{\eta}_{(r)}}\right\}.
\end{equation*}
While this method works well for smaller tables with sufficient sample size (e.g., $I_{k} = J = 3$, $n = 100$), it can struggle with larger and sparser tables (e.g., $I_{k} = J = 10$, $n = 50$). In these cases, the resulting threshold often leads to an inaccurate number of selected features (diverging significantly from the true model size). We propose an alternative method based on change-point detection. Similar to the maximum ratio method, we first sort the estimated statistics, i.e., $\widehat{\Delta}_{(K)}\geq \cdots \geq \widehat{\Delta}_{(1)}$. The tuning parameter $C_{1}$ can then be estimated as the point where the sorted sequence exhibits a significant change. Change-point detection methods like the two-piece linear model (conveniently implemented using R package \textit{segmented}) performs well in our simulations. 

The rationale behind this approach is that the distributions of $\{\widehat{\Delta}_{k}:~k\in S_{T}\}$ (features relevant to the response) and $\{\widehat{\Delta}_{k}:~k\in S\setminus S_{T}\}$ (irrelevant features) are expected to differ, with the former having stochastically greater values. Consequently, the sorted sequence often exhibits a change-point where the slope changes significantly. This method consistently provides a more stable estimate of the true model size $|S_{T}|$ compared to the ratio-based method, especially for larger and sparser tables (numerical examples are given in Section 3.4).  

\subsection{Simulation study}
We conducted a simulation study to assess the performance of the distance covariance screening under high-dimensional settings. In particular, we compare it with Pearson's Chi-squared screening by \cite{Huangetal} using ROC curves, as well as sensitivity and specificity based on selected tuning parameters by the change-point method that we proposed in Section 3.3. As discussed earlier, the maximum likelihood estimate ($\widehat{\Delta}$) and unbiased estimate ($\widetilde{\Delta}$) are asymptotically equivalent. For computational efficiency, we focus on $\widehat{\Delta}$ in this study, as our simulations indicate that $\widetilde{\Delta}$ has almost the same ROC curve. 

We consider the following simulation settings with dimension $|S| = K = 10,000$.
\begin{itemize}
\item Setting 1: $|S_{T}|/|S| = 5\%$, $(I_{k},~J) = (8,~ 8)$. For each $X_{k}\in S_{T}$, assign $\pi_{i_{k}j} = 1/20$ for 10 randomly selected $(i_{k},~j)$ pairs, and $\pi_{i_{k}j} = 1/108$ for the remaining pairs. For $X_{k}\in S\setminus S_{T}$, assign $\pi_{i_{k}j} = 1/64$. Sample sizes are $n = 25, ~50,~75,~100$ (average cell count ranges from 0.4 to 1.6).
\item Setting 2: $|S_{T}|/|S| = 5\%$, $(I_{k},~J) = (10,~ 10)$. For each $X_{k}\in S_{T}$, assign $\pi_{i_{k}j} = 1/50$ for 20 randomly selected $(i_{k},~j)$ pairs, and $\pi_{i_{k}j} = 1/150$ for the remaining pairs. For $X_{k}\in S\setminus S_{T}$, assign $\pi_{i_{k}j} = 1/100$. Sample sizes are $n = 25, ~50,~75,~100$ (average cell count ranges from 0.25 to 1).
\item Setting 3: Same as Setting 1, but $|S_{T}|/|S| = 10\%$. 
\item Setting 4: Same as Setting 2, but $|S_{T}|/|S| = 10\%$. 
\end{itemize}

Figures 1-4 depict the ROC curves for both distance covariance screening and Pearson's Chi-squared screening across all simulation settings. Table 1 summarizes the corresponding area under the ROC curve (AUC). Distance covariance screening consistently outperforms Pearson's Chi-squared screening, particularly for sparser tables (lower average cell counts). For example, in Setting 1 with a sample size of $25$ (resulting in an average cell count as low as 0.4), distance covariance achieves a significantly higher AUC of 0.776 compared to 0.658 for Pearson's Chi-squared screening. Similarly, in Setting 2 with $n=25$ (average cell count of only 0.25), the AUC values are 0.625 and 0.566 for distance covariance and Pearson's Chi-squared, respectively. As expected, both methods perform well with larger sample sizes (e.g., $n=100$).

Beyond AUC, we compared the sensitivity and specificity of both methods using a common tuning parameter selection. As mentioned in Section 3.3, the maximum ratio method by \cite{Huangetal} can be unstable for large sparse tables. In contrast, the change-point method consistently provides reasonable estimates for the true model size ($|S_{T}|$). For instance, in Setting 1 with the distance covariance screening and $n = 50$, the maximum ratio method estimated $\widehat{|S_{T}|} \approx 6,900$, while the change-point method (implemented by R function \textit{segmented}) yielded a much more accurate estimate of $\widehat{|S_{T}|} \approx 470$. Table 1 summarizes the sensitivity and specificity results. Here, distance covariance screening demonstrates consistent improvement for both metrics. For instance, in Setting 3 with $n=75$, the distance covariance screening achieved a sensitivity of 0.87 and a specificity of 0.981, compared to 0.791 and 0.979 by Pearson's Chi-squared screening. 

The simulation study also revealed limitations for both methods under extremely low average cell counts (high dimensionality and small sample size). In Setting 2 with $n=25$ (average cell count of 0.25), the sensitivities were as low as 0.17 (distance covariance) and 0.11 (Pearson's Chi-squared). This is due to the increased variances in the estimates ($\widehat{\Delta}_{k}$ and $\widehat{\eta}_{k}$). Consequently, it becomes more challenging to distinguish between relevant and irrelevant features using either method.

\begin{center}
[Figure 1 about here]
\end{center}

\begin{center}
[Figure 2 about here]
\end{center}

\begin{center}
[Figure 3 about here]
\end{center}

\begin{center}
[Figure 4 about here]
\end{center}

\begin{center}
[Table 1 about here]
\end{center}

\section{Approximate null distribution}
This section investigates hypothesis testing using sample distance correlation. As mentioned in Section 1, the two permutation tests outperform the traditional Pearson's Chi-squared test for large and sparse contingency tables. However, their computational cost is a major drawback due to the absence of an asymptotic theory for analytical p-value calculation. 

Here, we explore two approaches to approximate the null distribution of the bias-corrected distance correlation estimate. These include the Chi-squared approximation with one degree of freedom introduced by \cite{Shenetal} and a novel weighted Chi-squared approximation. For an $I\times J$ table, our weighted Chi-squared approximation requires calculating eigenvalues for two square matrices of dimensions $I\times I$ and $J\times J$, which is significantly more efficient than the $n\times n$ matrix computations required for continuous data. Notably, closed-form solutions for eigenvalues exist when $I$ and $J$ are less than or equal to 3. Our simulations demonstrate that both methods perform well for relatively large and sparse tables, with the Chi-squared method being slightly conservative.

We first introduce the bias-corrected estimate of distance correlation. The unbiased estimate of distance covariance is simply $\widetilde{\Delta}$. The unbiased estimate of squared distance variance $\widetilde{\Omega}$ is presented by Edelmann et al. (2020), and as a special case, we give the expression of $\widetilde{\Omega}$ for categorical variables. By Equations 2.1-2.7 in \cite{DStd} 
\begin{equation*}
\widetilde{\Omega} = \frac{T_{1}}{n(n-3)} - \frac{2T_{2}}{n(n-2)(n-3)} + \frac{T_{3}}{n(n-1)(n-2)(n-3)},
\end{equation*}
where
\begin{align*}
T_{1} & = \sum_{l=1}^{n}\sum_{m=1}^{n} \| X_{l} - X_{m} \|^2, \\
T_{2} & = \sum_{l=1}^{n} \sum_{m=1}^{n} \sum_{q=1}^{n} \| X_{l} - X_{m} \| \cdot \| X_{l} - X_{q} \| , \\
T_{3} & = \left(\sum_{l=1}^{n}\sum_{m=1}^{n} \| X_{l} - X_{m} \|\right)^2. 
\end{align*}

The expression of $\widetilde{\Omega}$ for categorical variables is given below (derivation is provided in A.4)
\begin{align*}
\widetilde{\Omega} = & \frac{n^3}{(n-1)(n-2)(n-3)}\left( 1-\sum_{i = 1}^{I}\widehat{\pi}_{i}^{2} \right)^2 -\frac{2n^2}{(n-2)(n-3)} \sum_{i = 1}^{I}\widehat{\pi}_{i}^{3}  \\
& - \frac{n(n-6)}{(n-2)(n-3)}\sum_{i = 1}^{I}\widehat{\pi}_{i}^{2} - \frac{n(n+2)}{(n-2)(n-3)}.
\end{align*}

By \cite{Shenetal, Zhangetal}, we have 
\begin{equation*}
\frac{n\widetilde{\Delta}(X, Y)}{\sqrt{\widetilde{\Omega}(X)\widetilde{\Omega}(Y)}}\longrightarrow \sum_{l=1}^{\infty}\sum_{m=1}^{\infty} \omega_{lm}(Z^{2}_{lm}-1),
\end{equation*}
as $n\rightarrow\infty$, where $Z_{lm}$'s are independent standard normal random variables, and the weights $\omega_{lm}$'s only depend on the marginal distributions of $X$ and $Y$. To compute $\omega_{lm}$'s, let $d(X_{l},~ X_{m})$ be a distance metric such that $d(X_{l},~ X_{m}) = 0$ if $X_{l} = X_{m}$ and $d(X_{l},~ X_{m}) = 1$ otherwise. Let $\mathbf{D}^{X}$ be the $n\times n$ distance matrix of $X$ with $\mathbf{D}^{X}_{lm} = d(X_{l},~ X_{m})$, and $\mathbf{H} = \mathbf{I} - \mathbf{J}/n$, where $\mathbf{I}$ represents the identity matrix and $\mathbf{J}$ the matrix of ones, $\{\lambda_{l}\}$ and $\{\mu_{m}\}$ are the limiting eigenvalues of $\mathbf{H}\mathbf{D}^{X}\mathbf{H}/n$ and $\mathbf{H}\mathbf{D}^{Y}\mathbf{H}/n$, respectively, and 
\begin{equation*}
\omega_{lm} = \frac{\lambda_{l}\mu_{m}}{\sqrt{\sum_{l=1}^{\infty}\lambda^{2}_{l}\sum_{m=1}^{\infty}\mu^{2}_{m}}}.
\end{equation*}
It can be shown that under $d(\cdot, \cdot)$, $\{\lambda_{l}\}$ and $\{\mu_{m}\}$ are all non-positive, therefore 
\begin{equation*}
\sum_{l=1}^{\infty}\sum_{m=1}^{\infty}\omega_{lm}  = 1
\end{equation*}
and $\sum_{l=1}^{\infty}\sum_{m=1}^{\infty} \omega_{lm}Z^{2}_{lm}$ is a weighted sum of Chi-squared distributions with $df=1$. For continuous data, calculating $\omega_{lm}$'s generally requires solving the eigenvalues of two $n\times n$ matrices, which can be computationally expensive. This cost, however, is significantly reduced for categorical variables, making the approximation feasible.

Since $\mathbf{H}$ is the centering matrix, it follows that $\mathbf{H}\mathbf{D}^{X}\mathbf{H}$ shares the same eigenvalues as $\mathbf{D}^{X}\mathbf{H}$. The matrices $\mathbf{D}^{X}$ and $\mathbf{D}^{X}\mathbf{H}$ can be simplified as follows 
\begin{equation*}
\mathbf{D}^{X} = \begin{bmatrix} 
    \mathbf{0}_{(n_{1+})\times (n_{1+})} & \mathbf{1}_{(n_{1+})\times (n_{2+})} & \dots & \mathbf{1}_{(n_{1+})\times (n_{I+})} \\
     \mathbf{1}_{(n_{2+})\times (n_{1+})} & \mathbf{0}_{(n_{2+})\times (n_{2+})} & \dots & \mathbf{1}_{(n_{2+})\times (n_{I+})} \\
    \vdots & \vdots  & \ddots & \vdots \\
    \mathbf{1}_{(n_{I+})\times (n_{1+})} &  \mathbf{1}_{(n_{I+})\times (n_{2+})} & \dots & \mathbf{0}_{(n_{I+})\times (n_{I+})}  
    \end{bmatrix},
 \end{equation*}   
 
 \begin{equation*}
\mathbf{D}^{X}\mathbf{H} = \begin{bmatrix} 
    (\widehat{\pi}_{1+}-1)\mathbf{1}_{(n_{1+})\times (n_{1+})} & \widehat{\pi}_{1+}\mathbf{1}_{(n_{1+})\times (n_{2+})} & \dots & \widehat{\pi}_{1+}\mathbf{1}_{(n_{1+})\times (n_{I+})} \\
    \widehat{\pi}_{2+} \mathbf{1}_{(n_{2+})\times (n_{1+})} & (\widehat{\pi}_{2+}-1)\mathbf{1}_{(n_{2+})\times (n_{2+})} & \dots & \widehat{\pi}_{2+}\mathbf{1}_{(n_{2+})\times (n_{I+})} \\
    \vdots & \vdots  & \ddots & \vdots \\
    \widehat{\pi}_{I+}\mathbf{1}_{(n_{I+})\times (n_{1+})} &  \widehat{\pi}_{I+}\mathbf{1}_{(n_{I+})\times (n_{2+})} & \dots & (\widehat{\pi}_{I+}-1)\mathbf{1}_{(n_{I+})\times (n_{I+})}  
    \end{bmatrix}.
 \end{equation*}   
 
We observe that $\mathbf{D}^{X}\mathbf{H}$ has a constant block structure, therefore it has the same eigenvalues as 
  \begin{equation*}
 \begin{bmatrix} 
    (\widehat{\pi}_{1+}-1)n_{1+} & \widehat{\pi}_{1+}\sqrt{n_{1+}n_{2+}} & \dots & \widehat{\pi}_{1+}\sqrt{n_{1+}n_{I+}} \\
    \widehat{\pi}_{2+} \sqrt{n_{2+}n_{1+}} & (\widehat{\pi}_{2+}-1)n_{2+} & \dots & \widehat{\pi}_{2+} \sqrt{n_{2+}n_{I+}} \\
    \vdots & \vdots  & \ddots & \vdots \\
    \widehat{\pi}_{I+}\sqrt{n_{I+}n_{1+}} &  \widehat{\pi}_{I+}\sqrt{n_{I+}n_{2+}} & \dots & (\widehat{\pi}_{I+}-1)n_{I+}  
    \end{bmatrix},
 \end{equation*}  
 
or equivalently

  \begin{equation*}
 n\begin{bmatrix} 
    (\widehat{\pi}_{1+}-1)\widehat{\pi}_{1+} & \widehat{\pi}_{1+}\sqrt{\widehat{\pi}_{1+}\widehat{\pi}_{2+}} & \dots & \widehat{\pi}_{1+}\sqrt{\widehat{\pi}_{1+}\widehat{\pi}_{I+}} \\
    \widehat{\pi}_{2+} \sqrt{\widehat{\pi}_{2+}\widehat{\pi}_{1+}} & (\widehat{\pi}_{2+}-1)\widehat{\pi}_{2+} & \dots & \widehat{\pi}_{2+} \sqrt{\widehat{\pi}_{2+}\widehat{\pi}_{I+}} \\
    \vdots & \vdots  & \ddots & \vdots \\
    \widehat{\pi}_{I+}\sqrt{\widehat{\pi}_{I+}\widehat{\pi}_{1+}} &  \widehat{\pi}_{I+}\sqrt{\widehat{\pi}_{I+}\widehat{\pi}_{2+}} & \dots & (\widehat{\pi}_{I+}-1)\widehat{\pi}_{I+} 
    \end{bmatrix},
 \end{equation*}  
 
therefore $\{\lambda_{l}\}$ are the eigenvalues of the following limiting matrix

  \begin{equation*}
 \begin{bmatrix} 
    (\pi_{1+}-1)\pi_{1+} & \pi_{1+}\sqrt{\pi_{1+}\pi_{2+}} & \dots & \pi_{1+}\sqrt{\pi_{1+}\pi_{I+}} \\
    \pi_{2+} \sqrt{\pi_{2+}\pi_{1+}} & (\pi_{2+}-1)\pi_{2+} & \dots & \pi_{2+} \sqrt{\pi_{2+}\pi_{I+}} \\
    \vdots & \vdots  & \ddots & \vdots \\
    \pi_{I+}\sqrt{\pi_{I+}\pi_{1+}} &  \pi_{I+}\sqrt{\pi_{I+}\pi_{2+}} & \dots & (\pi_{I+}-1)\pi_{I+} 
    \end{bmatrix}.
 \end{equation*}  
 
The matrix above has a rank of $I-1$ and $I-1$ non-zero eigenvalues. In the special case of $I=2$, we have $\lambda_{1} = -2\pi_{1+}\pi_{2+}$ and $\lambda_{i} = 0$ for $i\geq 2$. For $I = 3$, 
\begin{align*}
\lambda_{1} & = -(\pi_{1+}\pi_{2+}+\pi_{1+}\pi_{3+} + \pi_{2+}\pi_{3+}) - \sqrt{\pi_{1+}^{2}\pi_{2+}^{2}+\pi_{1+}^{2}\pi_{3+}^{2} + \pi_{2+}^{2}\pi_{3+}^{2}-\pi_{1+}\pi_{2+}\pi_{3+}}, \\
\lambda_{2} & = -(\pi_{1+}\pi_{2+}+\pi_{1+}\pi_{3+} + \pi_{2+}\pi_{3+}) + \sqrt{\pi_{1+}^{2}\pi_{2+}^{2}+\pi_{1+}^{2}\pi_{3+}^{2} + \pi_{2+}^{2}\pi_{3+}^{2}-\pi_{1+}\pi_{2+}\pi_{3+}}, 
\end{align*}
and $\lambda_{i} = 0$ for $i\geq 3$.

The asymptotic null distribution of $n\widetilde{\Delta}(X, Y)/\sqrt{\widetilde{\Omega}(X)\widetilde{\Omega}(Y)}$ can be then approximated by $\sum_{i=1}^{I-1}\sum_{j=1}^{J-1} \widehat{\omega}_{ij}Z^{2}_{ij}-1$, where $\widehat{\omega}_{ij}$ can be obtained from the estimated marginal probabilities, i.e., $\{\widehat{\pi}_{i+}\}$ and $\{\widehat{\pi}_{+j}\}$. We use simulated data to evaluate the performance of the two approximate null distribution, including Shen et al.'s Chi-squared approximation and our weighted Chi-squared approximation. We consider the following two simulation settings:

\begin{itemize}
\item Null setting 1: $(I,~J) = (8,~ 8)$. Assign $\pi_{ij} = 1/64$, $8\geq i\geq 1$, $8\geq j\geq 1$. Sample sizes are $n = 32, ~64,~96,~128$ (average cell counts are $0.5,~1,~1.5,~2$).
\item Null setting 2: $(I,~J) = (10,~ 10)$. Assign $\pi_{ij} = 1/100$, $10\geq i\geq 1$, $10\geq j\geq 1$. Sample sizes are $n = 50,~100,~150,~200$ (average cell counts are $0.5,~1,~1.5,~2$).
\end{itemize}

Figures 5 and 6 display the Quantile-Quantile (QQ) plots of the two approximations evaluated over $2,000$ simulations. These plots demonstrate that the weighted Chi-squared method accurately captures the null distribution of $n\widetilde{\Delta}(X, Y)/\sqrt{\widetilde{\Omega}(X)\widetilde{\Omega}(Y)}$, even for sparse tables (e.g., with average cell count as low as 0.5). 

We further compared the empirical type I error rates obtained from $20,000$ simulations. As shown in Figure 7, both methods maintain the desired type I error rate of 5\%. The Chi-squared method exhibits slight conservativeness. This, as Figures 5 and 6 suggest, is because the Chi-squared distribution with $df=1$ has a heavier right tail compared to the weighted chi-squared. However, it is important to note that the Chi-squared approximation's conservativeness is mild (average type I error rate around 0.039), resulting in satisfactory statistical power. 

\begin{center}
[Figure 5 about here]
\end{center}

\begin{center}
[Figure 6 about here]
\end{center}

\begin{center}
[Figure 7 about here]
\end{center}

Figure 8 explores the statistical power under two simulation settings similar to those in Section 3.3:
\begin{itemize}
\item Alternative setting 1: $(I,~J) = (8,~ 8)$. Assign $\pi_{ij} = 1/20$ for 10 randomly selected $(i,~j)$ pairs, and $\pi_{ij} = 1/108$ for the remaining pairs. Sample sizes are $n = 32, ~64,~96,~128$ (average cell counts are $0.5,~1,~1.5,~2$).
\item Alternative setting 2: $(I,~J) = (10,~ 10)$. Assign $\pi_{ij} = 1/50$ for 20 randomly selected $(i,~j)$ pairs, and $\pi_{ij} = 1/150$ for the remaining pairs. Sample sizes are $n = 50,~100,~150,~200$ (average cell counts are $0.5,~1,~1.5,~2$).
\end{itemize}

Here, our weighted Chi-squared test achieves statistical power nearly identical to the permutation test. The Chi-squared test also delivers satisfactory power, although its mild conservativeness leads to slightly lower power in all settings. For example, in alternative setting 1 with $n = 64$, the Chi-squared test exhibits a power of 84.6\%, compared to 87.5\% for the permutation test and 87.1\% for the weighted chi-squared test. Notably, for relatively large sample sizes, the three methods become almost equally powerful.

\begin{center}
[Figure 8 about here]
\end{center}

\section{Discussion and conclusions}
Detecting the association between two categorical variables is a durable research topic in statistics. The traditional Pearson's Chi-squared test can be unreliable for large and sparse tables due to violations of normality assumptions. Recent proposals, like the distance covariance permutation test and U-statistic permutation test, have demonstrably higher power under sparse conditions. Notably, both tests employ the distance covariance functional, but with slightly different estimators. While these tests have shown promise, the theoretical properties of the distance covariance for categorical data remain unclear. In this paper, we address this gap by investigating key statistical properties of this promising functional. Firstly, we demonstrate that distance covariance is a more robust measure of dependence compared to Pearson's Chi-squared, particularly for large and sparse tables. Secondly, we establish the strong screening consistency of this functional under mild assumptions, and showcase its effectiveness through simulations. Finally, we derive an accurate null distribution for a bias-corrected estimate of distance correlation, even for sparse tables. Our findings shed light on this valuable yet understudied dependence measure, facilitating its applications to large-scale data analysis.

We discuss here a potential extension of the presented method for conditional feature screening. In Section 3, we focused on marginal correlations. However, in practice, confounders often exist and should be controlled for when assessing the association between a predictor and response. If $W$ is a categorical confounder with $Q$ levels, the distance covariance can be extended for conditional feature screening as follows
\begin{equation*}
\Delta(X_{k}, Y|W) = \sum_{q = 1}^{Q}\sum_{i_{k} = 1}^{I_{k}}\sum_{j = 1}^{J} (\pi_{i_{k}j|q} - \pi_{i_{k}+|q}\pi_{+j|q})^{2},
\end{equation*}
where $\pi_{ij|q}$, $\pi_{i+|q}$, and $\pi_{+j|q}$ denote the conditional probabilities given level $q$ of $W$. Notably, $\Delta(X_{k}, Y|W) = 0$ if and only if $X_{k}$ is independent of $Y$ conditional on $W$. This property makes it suitable for conditional feature screening. The case of multiple confounders $(W_{1},~...,~W_{s})$ presents a greater challenge. Here, the sample size within each $X_{k}-Y$ partial table conditional on all confounders can be very small, leading to unreliable estimates of conditional probabilities. Addressing conditional feature screening with multiple confounders and limited data requires further investigation. 

Nevertheless, when sample size permits, the conditional distance covariance can be readily applied to conditional feature screening. An immediate application is the network learning with categorical variables. Let $(W_{1},~...,~W_{s})$ represent a subset of the predictors that are correlated with $X_{k}$ or $Y$. Under the assumption of Markov property, a normalized version of $\Delta(X_{k},~ Y|W_{1},~...,~W_{s})$ could be used as a correlation metric within existing network learning algorithms like the PC algorithm \cite{PC}. However, this requires substantial theoretical and numerical exploration, which we leave for future research.

\section*{Competing Interests}
\noindent
The author has declared that no competing interests exist.

\section*{Acknowledgement}
\noindent
The work was supported by an NSF DBI Biology Integration Institute (BII) grant (award no. 2119968).

\section*{Appendix}
\subsection*{A.1. Proof of Lemma 1}
To derive $\mbox{IF}[(x,~ y),~ \Delta,~ \pi)]$, we divide the cells into four groups, namely $(i\neq x, j\neq y)$, $(i = x, j\neq y)$, $(i\neq x, j = y)$, and $(i = x, j = y)$. For cell $(i,~ j)$, let $\Delta_{ij} = (\pi_{ij} - \pi_{i+}\pi_{+j})^2$ and $\Delta = \sum_{i=1}^{I}\sum_{j=1}^{J}\Delta_{ij}$. After the $\epsilon$-modifications, the joint probabilities become $(1-\epsilon)\pi_{xy} + \epsilon$ for $(x, y)$ and $(1-\epsilon)\pi_{ij}$ for the remaining cells. 

Let $\Delta_{ij\epsilon}$ be the value of $\Delta_{ij}$ after the $\epsilon$-modifications. For $(i\neq x, j\neq y)$, we have 
\begin{align*}
& \Delta_{ij\epsilon} = \left[(1-\epsilon)\pi_{ij} - (1-\epsilon)^{2}\pi_{i+}\pi_{+j}\right]^{2}, \\
& \Delta_{ij\epsilon} - \Delta_{ij} = \epsilon\left[ (2-\epsilon)\pi_{ij} - \pi_{i+}\pi_{+j} - (1-\epsilon)^{2}\pi_{i+}\pi_{+j}\right] \left[ (2-\epsilon)\pi_{i+}\pi_{+j} -\pi_{ij}  \right], \\
& \lim_{\epsilon\rightarrow 0} \frac{\Delta_{ij\epsilon} - \Delta_{ij}}{\epsilon} = 2(\pi_{ij} - \pi_{i+}\pi_{+j})(2\pi_{i+}\pi_{+j} -\pi_{ij}).
\end{align*}
For $(i = x, j\neq y)$, we have 
\begin{align*}
& \Delta_{xj\epsilon} = \left[(1-\epsilon)\pi_{xj} - (1-\epsilon)^{2}\pi_{x+}\pi_{+j} - \epsilon(1-\epsilon)\pi_{+j} \right]^{2}, \\
& \lim_{\epsilon\rightarrow 0} \frac{\Delta_{xj\epsilon} - \Delta_{xj}}{\epsilon} = 2(\pi_{xj} - \pi_{x+}\pi_{+j})(2\pi_{x+}\pi_{+j} -\pi_{xj} - \pi_{+j}).
\end{align*}
For $(i \neq x, j = y)$, we have 
\begin{align*}
& \Delta_{iy\epsilon} = \left[(1-\epsilon)\pi_{iy} - (1-\epsilon)^{2}\pi_{i+}\pi_{+y} - \epsilon(1-\epsilon)\pi_{+y} \right]^{2}, \\
& \lim_{\epsilon\rightarrow 0} \frac{\Delta_{iy\epsilon} - \Delta_{iy}}{\epsilon} = 2(\pi_{iy} - \pi_{i+}\pi_{+y})(2\pi_{i+}\pi_{+y} -\pi_{iy} - \pi_{i+}).
\end{align*}
For $(i = x, j = y)$, we have 
\begin{align*}
& \Delta_{xy\epsilon} = \left\{ (1-\epsilon)\pi_{xy} + \epsilon - [(1-\epsilon)\pi_{x+} + \epsilon][(1-\epsilon)\pi_{+y} + \epsilon] \right\}^{2}, \\
& \lim_{\epsilon\rightarrow 0} \frac{\Delta_{xy\epsilon} - \Delta_{xy}}{\epsilon} = 2(\pi_{xy} - \pi_{x+}\pi_{+y})(2\pi_{x+}\pi_{+y} - \pi_{x+} - \pi_{+y} - \pi_{xy} +1).
\end{align*}
Summarizing the results above, we have 
\begin{align*}
\mbox{IF}[(x,~ y),~ \Delta,~ \pi)] = & \sum_{i\neq x}\sum_{j\neq y}\lim_{\epsilon\rightarrow 0} \frac{\Delta_{ij\epsilon} - \Delta_{ij}}{\epsilon} + \sum_{i\neq x}\lim_{\epsilon\rightarrow 0} \frac{\Delta_{iy\epsilon} - \Delta_{iy}}{\epsilon} +\sum_{j\neq y}\lim_{\epsilon\rightarrow 0} \frac{\Delta_{xj\epsilon} - \Delta_{xj}}{\epsilon} + \lim_{\epsilon\rightarrow 0} \frac{\Delta_{xy\epsilon} - \Delta_{xy}}{\epsilon} \\
 = & 2\sum_{i=1}^{I}\sum_{j=1}^{J} (\pi_{ij} - \pi_{i+}\pi_{+j})(2\pi_{i+}\pi_{+j} - \pi_{ij}) - 2\sum_{j=1}^{J} \pi_{+j}(\pi_{xj} - \pi_{x+}\pi_{+j}) - 2\sum_{i=1}^{I} \pi_{i+}(\pi_{iy} - \pi_{i+}\pi_{+y}) \\
& + 2(\pi_{xy} - \pi_{x+}\pi_{+y}).
\end{align*}

\subsection*{A.2. Proof of Lemma 2}
Similar to the proof of Lemma 1, we divide the cells into four groups, $(i\neq x, j\neq y)$, $(i = x, j\neq y)$, $(i\neq x, j = y)$, and $(i = x, j = y)$. Let $\eta_{ij} = (\pi_{ij} - \pi_{i+}\pi_{+j})^{2}/(\pi_{i+}\pi_{+j})$, and $\eta_{ij\epsilon}$ be the value of $\eta_{ij}$ after the $\epsilon$-modifications. For $(i\neq x, j\neq y)$, we have
\begin{align*}
& \eta_{ij\epsilon} = \frac{[(1-\epsilon)\pi_{ij} - (1-\epsilon)^{2}\pi_{i+}\pi_{+j} ]^2}{(1-\epsilon)^{2}\pi_{i+}\pi_{+j}}, \\
& \eta_{ij\epsilon} - \eta_{ij} = 2\epsilon\left[  \pi_{ij} - (2-\epsilon)\pi_{i+}\pi_{+j} \right], \\
& \lim_{\epsilon\rightarrow 0} \frac{\eta_{ij\epsilon} - \eta_{ij}}{\epsilon} = 2\left(  \pi_{ij} - 2\pi_{i+}\pi_{+j} \right).
\end{align*}
For $(i = x, j = y)$, we have 
\begin{align*}
& \eta_{xy\epsilon} = \frac{\left\{ (1-\epsilon)\pi_{xy} + \epsilon - [(1-\epsilon)\pi_{x+} + \epsilon][(1-\epsilon)\pi_{+y} + \epsilon] \right\}^{2}}{[(1-\epsilon)\pi_{x+} + \epsilon][(1-\epsilon)\pi_{+y} + \epsilon]}, \\
& \lim_{\epsilon\rightarrow 0} \frac{\eta_{xy\epsilon} - \eta_{xy}}{\epsilon} = \frac{\pi_{xy}^{2}}{\pi_{x+}\pi_{+y}^2} + \frac{\pi_{xy}^{2}}{\pi_{x+}^{2}\pi_{+y}} - \frac{2\pi_{xy}}{\pi_{x+}\pi_{+y}} - 2\pi_{xy} + 2\pi_{x+}\pi_{+y} - \pi_{x+} - \pi_{+y} -2.
\end{align*}
For $(i = x, j\neq y)$ and $(i \neq x, j = y)$, we have
\begin{align*}
& \lim_{\epsilon\rightarrow 0} \frac{\eta_{xj\epsilon} - \eta_{xj}}{\epsilon} =  - \frac{\pi_{xj}^{2}}{\pi_{x+}\pi_{+j}} - \frac{\pi_{xj}^{2}}{\pi_{x+}^{2}\pi_{+j}} - 3\pi_{x+} + 3\pi_{+j} + 4\pi_{xj}, \\
& \lim_{\epsilon\rightarrow 0} \frac{\eta_{iy\epsilon} - \eta_{iy}}{\epsilon} =  - \frac{\pi_{iy}^{2}}{\pi_{i+}\pi_{+y}} - \frac{\pi_{iy}^{2}}{\pi_{i+}\pi_{+y}^{2}} + 3\pi_{i+} - 3\pi_{+y} + 4\pi_{iy}.
\end{align*}
This completes the proof.

\subsection*{A.3. Proof of Theorem 3}
\begin{proof}
We first show that $P[\widehat{S}(C_{1}) = S_{T}] \rightarrow 1$ as $n\rightarrow\infty$. The first step is to establish the uniform consistency of $\widehat{\Delta}_{k}$. Note that
\begin{align*}
\left| \widehat{\Delta}_{k} - \Delta_{k} \right| & = \left| \sum_{i_{k} = 1}^{I_{k}} \sum_{j=1}^{J} (\widehat{\pi}_{i_{k}j} - \widehat{\pi}_{i_{k}+}\widehat{\pi}_{+j})^{2} - \sum_{i_{k} = 1}^{I_{k}} \sum_{j=1}^{J}(\pi_{i_{k}j} - \pi_{i_{k}+}\pi_{+j})^{2} \right| \\
& < 2\left|  \sum_{i_{k} = 1}^{I_{k}} \sum_{j=1}^{J} (\widehat{\pi}_{i_{k}j} - \widehat{\pi}_{i_{k}+}\widehat{\pi}_{+j}) - \sum_{i_{k} = 1}^{I_{k}} \sum_{j=1}^{J}(\pi_{i_{k}j} - \pi_{i_{k}+}\pi_{+j}) \right| \\
& \leq 2\sum_{i_{k} = 1}^{I_{k}} \sum_{j=1}^{J} \left| \widehat{\pi}_{i_{k}j} - \pi_{i_{k}j} \right| + 2\sum_{i_{k} = 1}^{I_{k}} \sum_{j=1}^{J} \left| \widehat{\pi}_{i_{k}+}\widehat{\pi}_{+j} - \pi_{i_{k}+}\pi_{+j}\right|  \\
& \leq 2\sum_{i_{k} = 1}^{I_{k}} \sum_{j=1}^{J} \left| \widehat{\pi}_{i_{k}j} - \pi_{i_{k}j} \right| + 2\sum_{i_{k} = 1}^{I_{k}} \sum_{j=1}^{J} \left| \widehat{\pi}_{i_{k}+}\widehat{\pi}_{+j}  - \pi_{i_{k}+}\widehat{\pi}_{+j} + \pi_{i_{k}+}\widehat{\pi}_{+j} - \pi_{i_{k}+}\pi_{+j}\right| \\
& \leq 2 \sum_{i_{k} = 1}^{I_{k}} \sum_{j=1}^{J} \left| \widehat{\pi}_{i_{k}j} - \pi_{i_{k}j} \right| + 2\sum_{i_{k} = 1}^{I_{k}} \left| \widehat{\pi}_{i_{k}} - \pi_{i_{k}} \right| + 2\sum_{j=1}^{J} \left| \widehat{\pi}_{j} - \pi_{j} \right| \\
& \leq 6 \sum_{i_{k} = 1}^{I_{k}} \sum_{j=1}^{J} \left| \widehat{\pi}_{i_{k}j} - \pi_{i_{k}j} \right|.
\end{align*} 
It suffices to show the uniform consistency of $\widehat{\pi}_{i_{k}j}$. Similar to Huang et al. (2014), define $Z_{mi_{k}j} = I(X_{mk} = i_{k}) - \pi_{i_{k}j}$, where $m = 1,~...,~n$. It is easy to see that $|Z_{mi_{k}j}|<1$, $E(Z_{mi_{k}j}) = 0$ and $E(Z^{2}_{mi_{k}j}) = \pi_{i_{k}j} - \pi_{i_{k}j}^{2}$. By Bernstein's inequality, for any $\epsilon>0$, we have 
\begin{equation*}
P\left( \left| \frac{1}{n}\sum_{m=1}^{n} Z_{mi_{k}j} \right|>\epsilon \right) \leq 2\exp\left( -\frac{6n\epsilon^2}{4\epsilon + 3} \right).
\end{equation*} 
Note that $\widehat{\pi}_{i_{k}j} - \pi_{i_{k}j} = (1/n)\sum_{m=1}^{n} Z_{mi_{k}j} $, we have
\begin{align*}
P\left( \max_{j}\max_{k}\max_{i_{k}}\left| \widehat{\pi}_{i_{k}j} - \pi_{i_{k}j} \right|>\epsilon \right) & \leq \sum_{j=1}^{J}\sum_{k=1}^{K}\sum_{i_{k}=1}^{I_{k}} P\left( \frac{1}{n}\left| \sum_{m=1}^{n} Z_{mi_{k}j} \right| > \epsilon \right)  \\
& \leq  2K I_{max} J\cdot\exp\left( -\frac{6n\epsilon^2}{4\epsilon + 3} \right) \\
& \leq 2I_{max} J\cdot\exp\left( \log K-\frac{6n\epsilon^2}{4\epsilon + 3} \right). 
\end{align*}
By Condition (3), $\widehat{\pi}_{i_{k}j}$ is uniformly consistent, therefore $\widehat{\Delta}_{k}$ is uniformly consistent. By Condition (2), $\min_{k}\max_{i_{k}, j} \omega^{2}_{i_{k}j} > \omega^{2}_{min}$, therefore we have $\min_{k}\Delta_{k} > 2J\omega^{2}_{min}$ for $k\in S_{T}$. Let $C_{1} = J\omega^{2}_{min}$, by the uniform consistency of $\widehat{\Delta}_{k}$, we have $S_{T}\subset\widehat{S}(C_{1})$ almost surely. 

Next, we show $\widehat{S}(C_{1})\subset S_{T}$ almost surely. Suppose there exists $k^{*}\in \widehat{S}(C_{1})$, but  $k^{*}\notin S_{T}$, then $\widehat{\Delta}_{k^{*}}>J\omega^{2}_{min}$ and $\Delta_{k^{*}} = 0$. For $k^{*}$, we have $|\widehat{\Delta}_{k^{*}}-\Delta_{k^{*}}|>J\omega^{2}_{min}$. Let $\epsilon = J\omega^{2}_{min}$, we have
\begin{equation*}
P[S_{T}\not\subset\widehat{S}(C_{1})] \leq P(\max_{k}|\widehat{\Delta}_{k} - \Delta_{k}| > \epsilon) \rightarrow 0,
\end{equation*}
as $n\rightarrow\infty$, leading to a contradiction, therefore $\widehat{S}(C_{1})\subset S_{T}$ almost surely, and $P[\widehat{S}(C_{1}) = S_{T}]\rightarrow 1$, as $n\rightarrow\infty$.

The strong screening consistency of $\widetilde{S}(C_{2})$ can be established in the same way. It suffices to show that $\widetilde{\Delta}_{k}$ is uniformly consistent. In the definition of $\widetilde{\Delta}_{k}$, we can see 
\begin{align*}
\sum_{i_{k} = 1}^{I_{k}} \sum_{j=1}^{J} (\widehat{\pi}_{i_{k}j} - \widehat{\pi}_{i_{k}+}\widehat{\pi}_{+j})^{2} & \leq 2, \\
\sum_{i_{k} = 1}^{I_{k}} \sum_{j=1}^{J}\widehat{\pi}_{i_{k}j}\widehat{\pi}_{i_{k}+}\widehat{\pi}_{+j} & \leq 1,  \\
\sum_{i_{k} = 1}^{I_{k}}\widehat{\pi}_{i_{k}+}^{2} + \sum_{j=1}^{J}\widehat{\pi}_{+j}^{2} & \leq 2, \\ 
\left(\sum_{i_{k} = 1}^{I_{k}}\widehat{\pi}_{i_{k}+}^{2}\right)\left(\sum_{j=1}^{J}\widehat{\pi}_{+j}^{2}\right) & \leq 1.  
\end{align*}
Therefore $\widetilde{\Delta}_{k} = \widehat{\Delta}_{k} + O(1/n)$, and $|\widetilde{\Delta}_{k} - \Delta_{k}|\leq |\widehat{\Delta}_{k} - \Delta_{k}| + O(1/n)$, indicating that $\widetilde{\Delta}_{k}$ is also uniformly consistent.  This completes the proof.
\end{proof}

\subsection*{A.4. Derivation of $\widetilde{\Omega}$}
Let $n_{(m)} = \sum_{l=1}^{n}I(X_{l} = X_{m})$ for $m = 1,~...,~ n$, then we have 
\begin{align*}
T_{1} & = \sum_{m=1}^{n}(n-n_{(m)}) = n^{2} - \sum_{i=1}^{I}n_{i+}^{2} ,\\
T_{2} & = \sum_{m=1}^{n}(n-n_{(m)})^{2} = n^{3} -2n\sum_{i=1}^{I}n_{i+}^{2} + \sum_{i=1}^{I}n_{i+}^{3}, \\
T_{3} & =  \left( n^{2} - \sum_{i=1}^{I}n_{i+}^{2} \right)^{2} = n^{4} + \left( \sum_{i=1}^{I}n_{i+}^{2} \right)^2 - 2n^{2}\sum_{i=1}^{I}n_{i+}^{2}.
\end{align*}
Further we have 
\begin{align*}
\frac{T_{1}}{n(n-3)} & = \frac{n}{n-3} - \frac{n}{n-3}\sum_{i=1}^{I}\widehat{\pi}_{i+}^{2}, \\
\frac{2T_{2}}{n(n-2)(n-3)} & = \frac{2n^{2}}{(n-2)(n-3)} - \frac{4n^{2}}{(n-2)(n-3)}\sum_{i=1}^{I}\widehat{\pi}_{i+}^{2} + \frac{2n^{2}}{(n-2)(n-3)}\sum_{i=1}^{I}\widehat{\pi}_{i+}^{3}, \\
\frac{T_{3}}{n(n-1)(n-2)(n-3)} & = \frac{n^{3}}{(n-1)(n-2)(n-3)}\left( 1- \sum_{i=1}^{I}\widehat{\pi}_{i+}^{2} \right)^2.
\end{align*}
As $\widetilde{\Omega} = \frac{T_{1}}{n(n-3)} - \frac{2T_{2}}{n(n-2)(n-3)} + \frac{T_{3}}{n(n-1)(n-2)(n-3)}$, we have
\begin{align*}
\widetilde{\Omega} = & \frac{n^3}{(n-1)(n-2)(n-3)}\left( 1-\sum_{i = 1}^{I}\widehat{\pi}_{i+}^{2} \right)^2 -\frac{2n^2}{(n-2)(n-3)} \sum_{i = 1}^{I}\widehat{\pi}_{i+}^{3}  \\
& - \frac{n(n-6)}{(n-2)(n-3)}\sum_{i = 1}^{I}\widehat{\pi}_{i+}^{2} - \frac{2n(n+2)}{(n-2)(n-3)}.
\end{align*}

\newpage
\section*{Tables and Figures}
\begin{table}[ht]
\centering
\begin{tabular}{c|c|c|c|c|c|c}
  \hline  \hline
$|S_{T}|/|S|$  & $(I_{k}, J)$ & sample size & AUC-$\hat{\Delta}$ & AUC-$\hat{\eta}$ & (sensitivity, specificity)-$\hat{\Delta}$ & (sensitivity, specificity)-$\hat{\eta}$  \\ 
  \hline
5\% & (8, 8) & 25 & 0.776 & 0.658 & (0.362, 0.948) & (0.216, 0.947) \\
5\% & (8, 8) & 50 & 0.939 & 0.919 & (0.652, 0.982) & (0.598, 0.972) \\
5\% & (8, 8) & 75 & 0.985 & 0.974 & (0.854, 0.986) & (0.788, 0.981) \\
5\% & (8, 8) & 100 & 0.997 & 0.991 & (0.956, 0.990) & (0.893, 0.787) \\
\hline
5\% & (10, 10) & 25 & 0.625 & 0.566 & (0.170, 0.930) & (0.112, 0.924) \\
5\% & (10, 10) & 50 & 0.718 & 0.680 & (0.280, 0.939) & (0.218, 0.939) \\
5\% & (10, 10) & 75 & 0.811 & 0.783 & (0.386, 0.948) & (0.320, 0.945) \\
5\% & (10, 10) & 100 & 0.892 & 0.874 & (0.538, 0.963) & (0.484, 0.957) \\
  \hline
10\% & (8, 8) & 25 & 0.783 & 0.662 & (0.339, 0.961) & (0.212, 0.956) \\
10\% & (8, 8) & 50 & 0.939 & 0.912 & (0.634, 0.982) & (0.558, 0.971) \\
10\% & (8, 8) & 75 & 0.982 & 0.974 & (0.870, 0.981) & (0.791, 0.979) \\
10\% & (8, 8) & 100 & 0.995 & 0.991 & (0.965, 0.976) & (0.943, 0.976) \\
\hline
10\% & (10, 10) & 25 & 0.612 & 0.542 & (0.163, 0.934) & (0.106, 0.931) \\
10\% & (10, 10) & 50 & 0.739 & 0.704 & (0.247, 0.956) & (0.203, 0.949) \\
10\% & (10, 10) & 75 & 0.827 & 0.803 & (0.408, 0.955) & (0.359, 0.950) \\
10\% & (10, 10) & 100 & 0.877 & 0.863 & (0.447, 0.971) & (0.408, 0.969) \\
   \hline  \hline
\end{tabular}
\caption{AUC, sensitivity and specificity of distance covariance and Pearson's Chi-squared screenings, where the sensitivity and specificity are based on change-point selected tuning parameters.}
\end{table}

\newpage
\begin{figure}[!htbp]
\begin{center}
\includegraphics[scale=0.5]{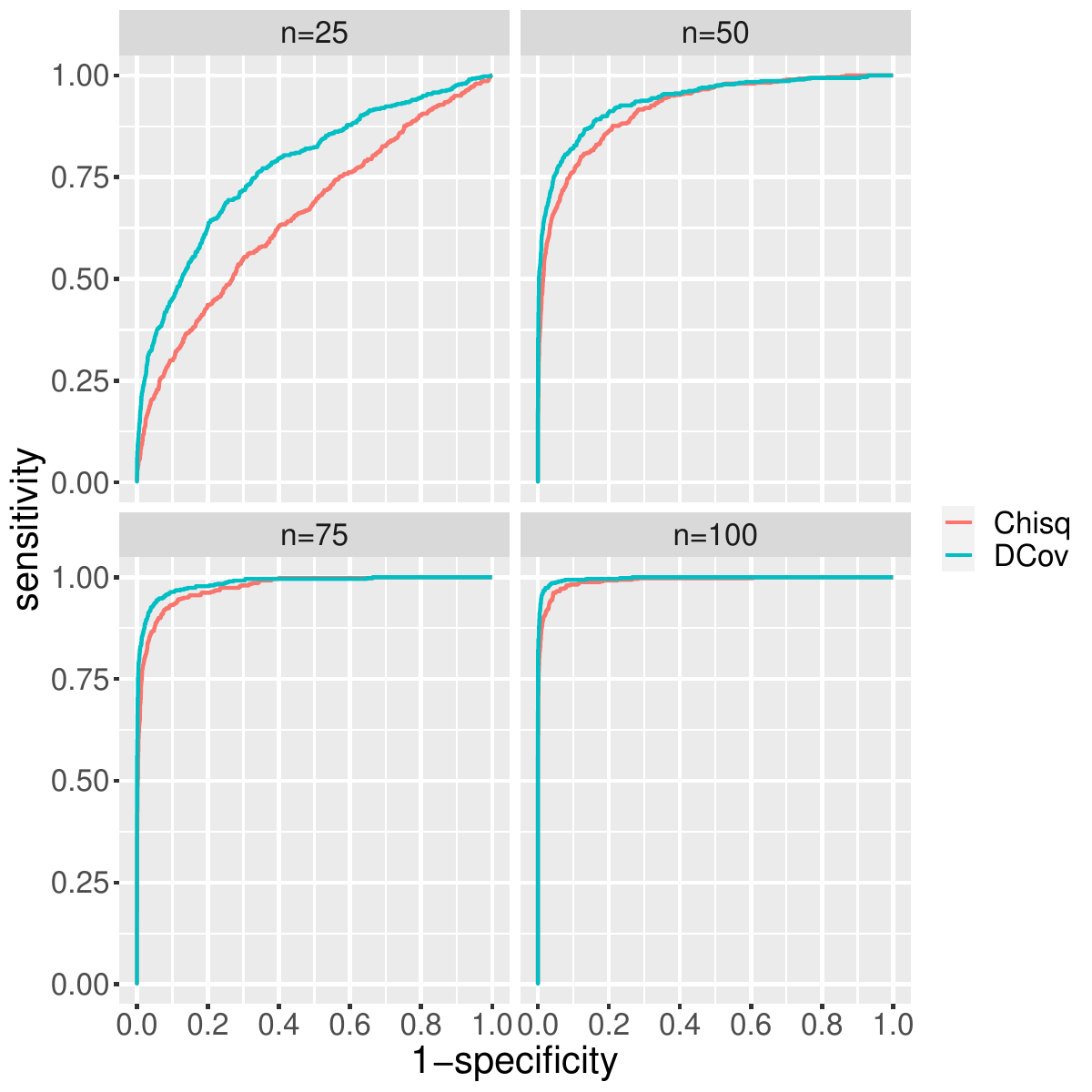}
\end{center}
\caption{ROC curves for distance covariance and Pearson's Chi-squared screenings under simulation setting 1 ($I_{k} = J = 8$ and $|S_{T}|/|S| = 5\%$).
}
\end{figure}

\newpage
\begin{figure}[!htbp]
\begin{center}
\includegraphics[scale=0.5]{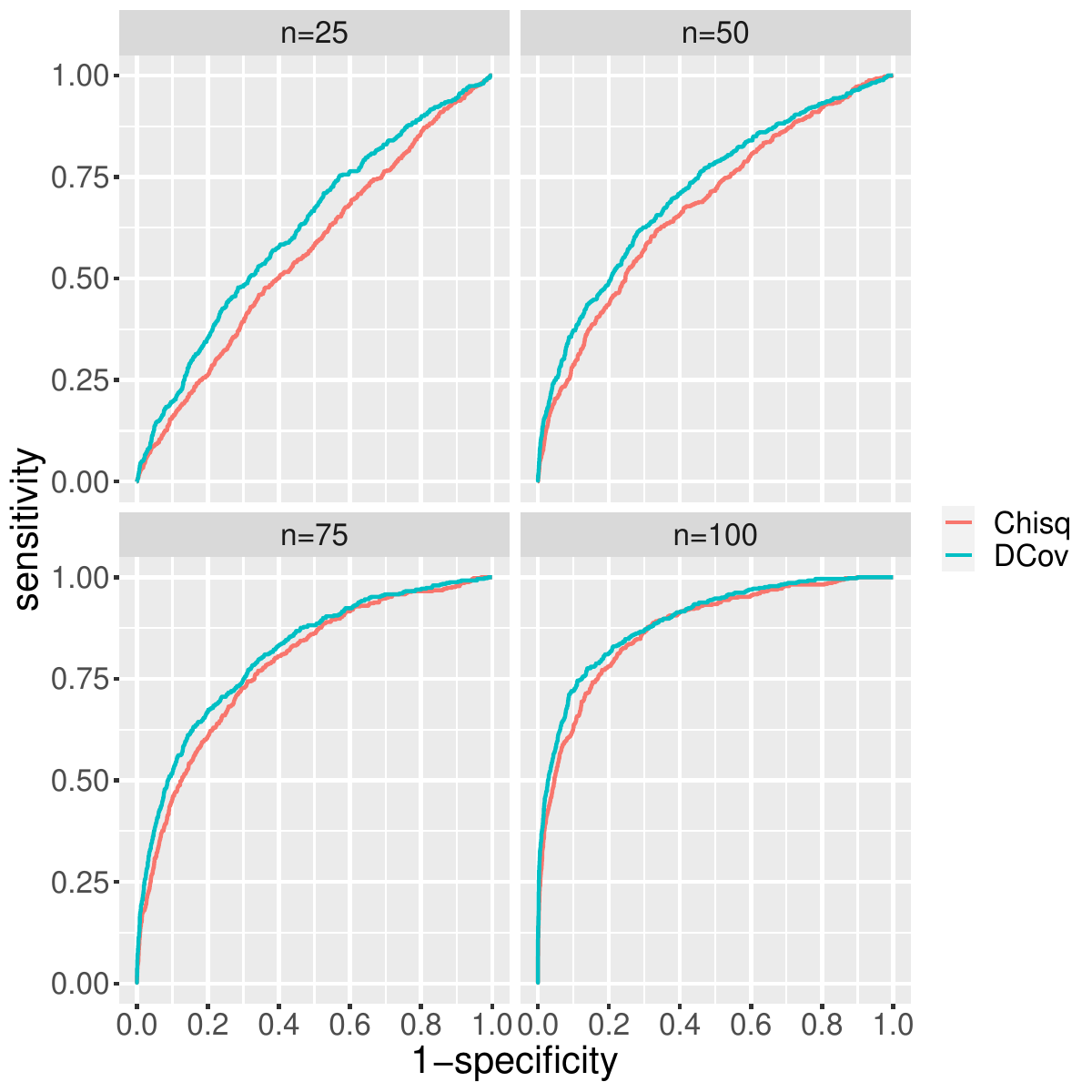}
\end{center}
\caption{ROC curves for distance covariance and Pearson's Chi-squared screenings under simulation setting 2 ($I_{k} = J = 10$ and $|S_{T}|/|S| = 5\%$).
}
\end{figure}

\newpage
\begin{figure}[!htbp]
\begin{center}
\includegraphics[scale=0.5]{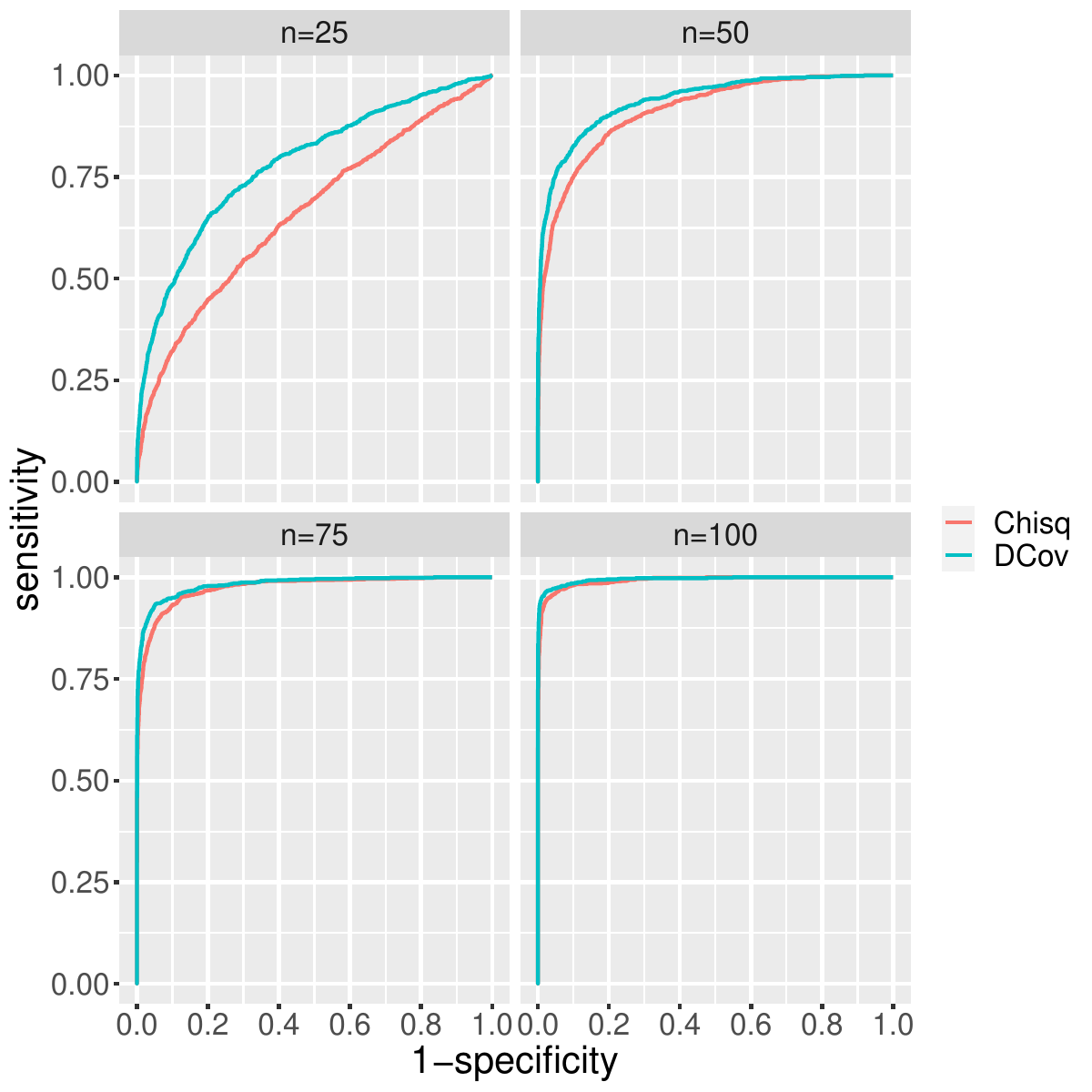}
\end{center}
\caption{ROC curves for distance covariance and Pearson's Chi-squared screenings under simulation setting 3 ($I_{k} = J = 8$ and $|S_{T}|/|S| = 10\%$).
}
\end{figure}

\newpage
\begin{figure}[!htbp]
\begin{center}
\includegraphics[scale=0.5]{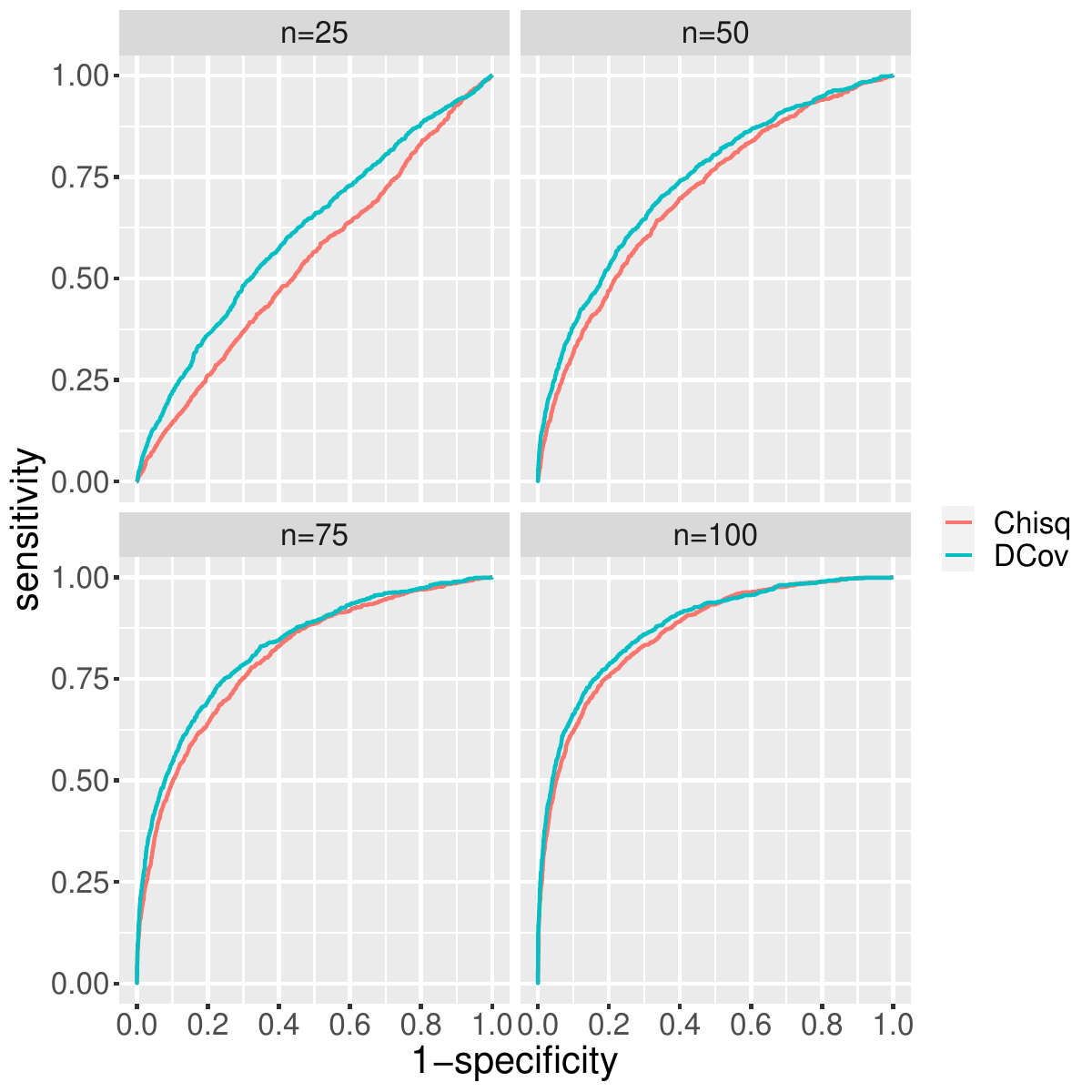}
\end{center}
\caption{ROC curves for distance covariance and Pearson's Chi-squared screenings under simulation setting 4 ($I_{k} = J = 10$ and $|S_{T}|/|S| = 10\%$).
}
\end{figure}

\newpage
\begin{figure}[!htbp]
\begin{center}
\includegraphics[scale=0.5]{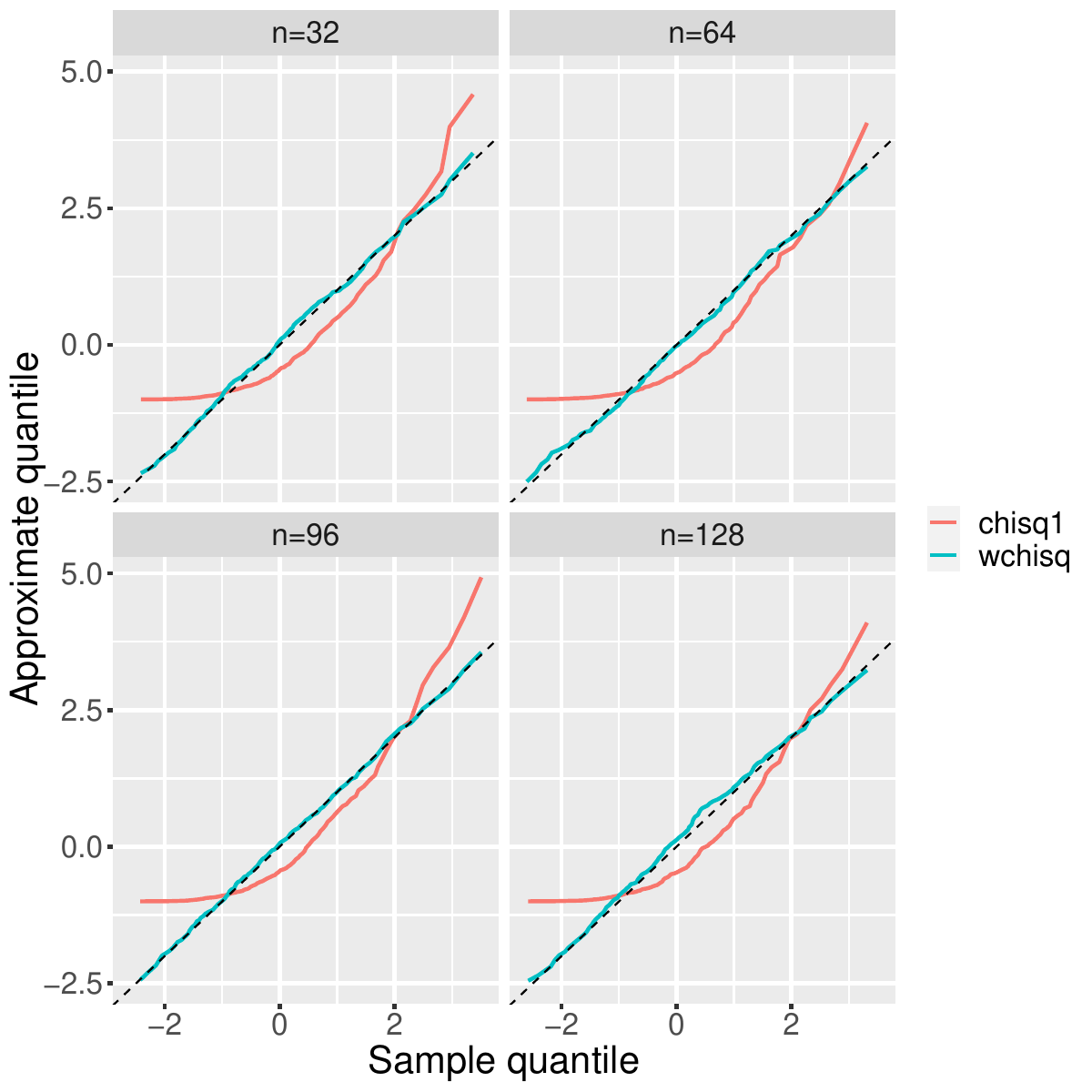}
\end{center}
\caption{QQ-plots comparing the weighted Chi-squared and Chi-squared ($df=1$) approximations to the null distribution for $I = J = 8$.
}
\end{figure}

\newpage
\begin{figure}[!htbp]
\begin{center}
\includegraphics[scale=0.5]{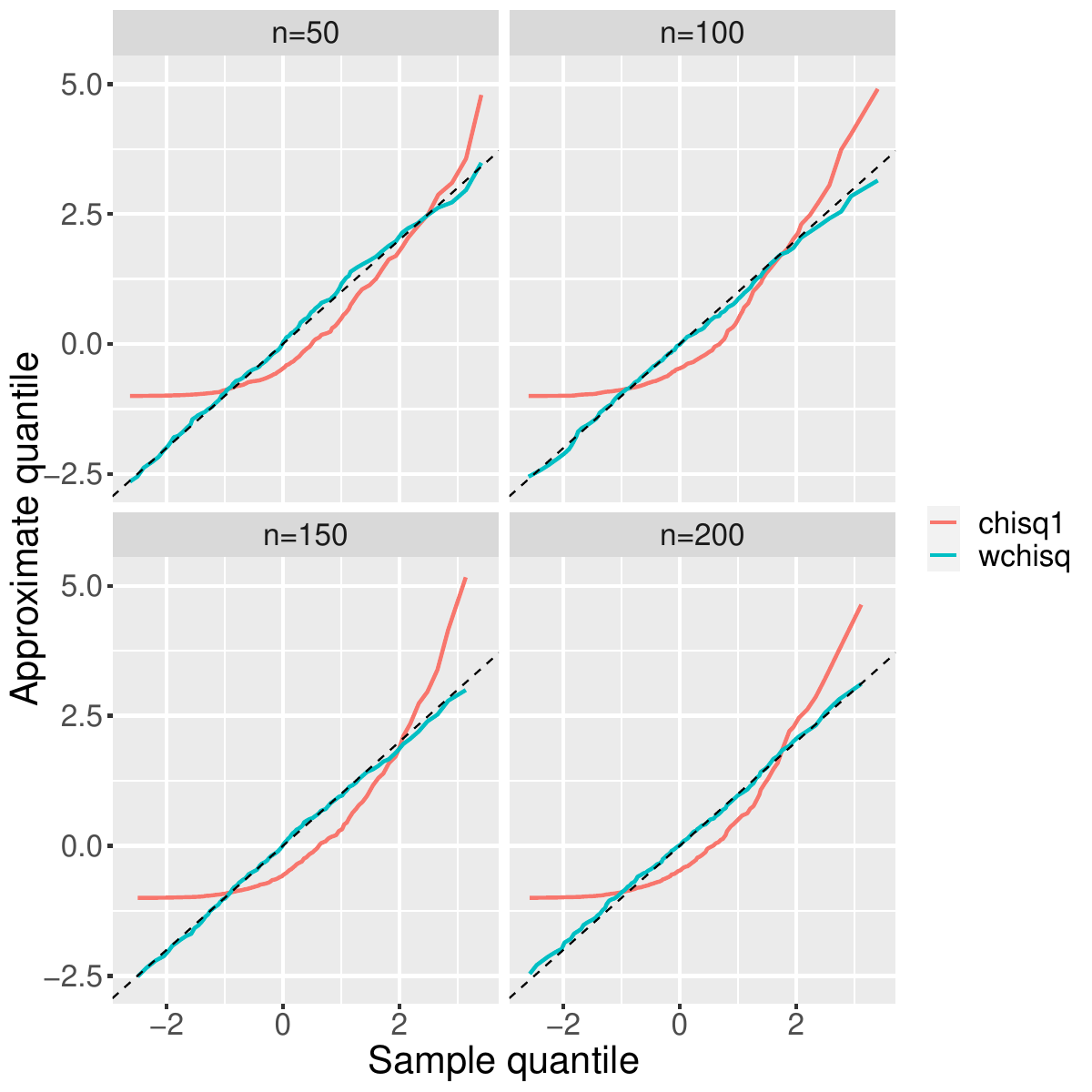}
\end{center}
\caption{QQ-plots comparing the weighted Chi-squared and Chi-squared ($df=1$) approximations to the null distribution for $I = J = 10$.
}
\end{figure}

\newpage
\begin{figure}[!htbp]
\begin{center}
\includegraphics[scale=0.5]{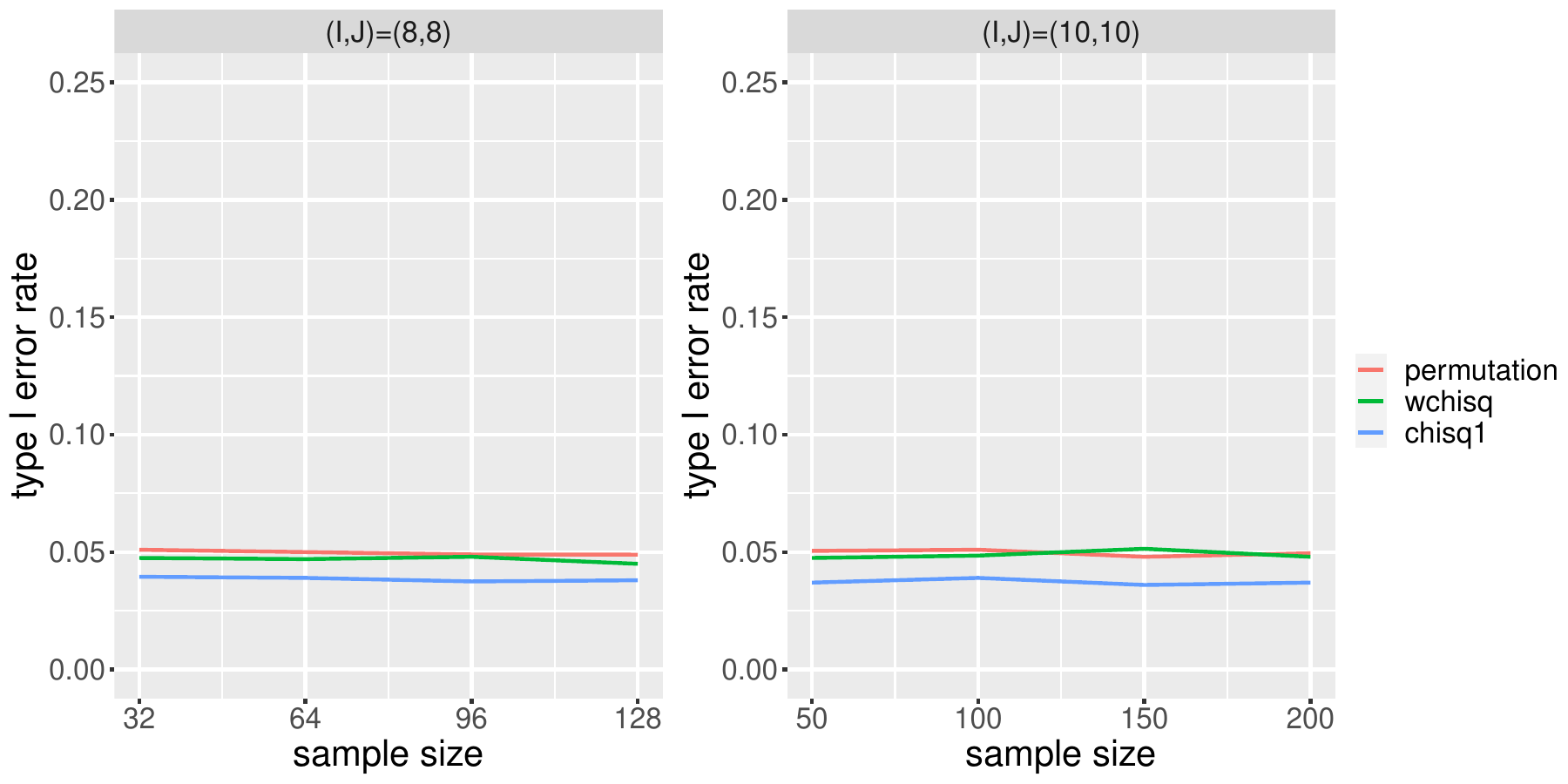}
\end{center}
\caption{Type I error rate of permutation test, weighted Chi-squared test, and Chi-squared test ($df=1$).
}
\end{figure}

\newpage
\begin{figure}[!htbp]
\begin{center}
\includegraphics[scale=0.5]{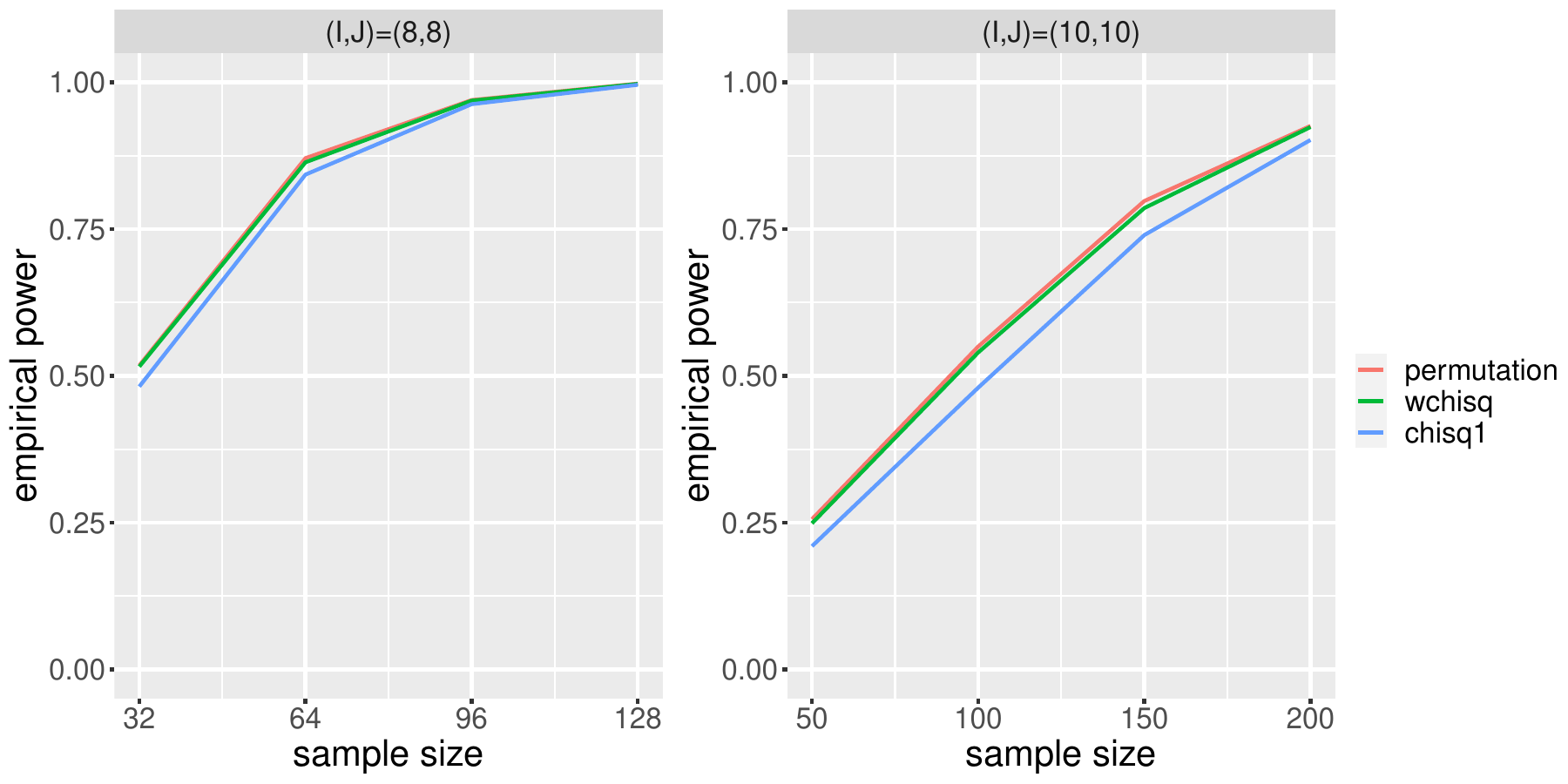}
\end{center}
\caption{Statistical power of permutation test, weighted Chi-squared test, and Chi-squared test ($df=1$).
}
\end{figure}

\end{document}